\documentclass{article}
\usepackage[utf8]{inputenc}
\usepackage{amsthm, amssymb, amsmath, thmtools}
\usepackage{url}
\usepackage[toc,page]{appendix}
\usepackage[numbers,sort&compress]{natbib}
\usepackage{geometry}
\usepackage{caption}
\usepackage{subcaption}
\usepackage{float}
\usepackage{array}
\usepackage{enumitem}
\usepackage{footnote}
\makesavenoteenv{tabular}
\makesavenoteenv{table}

\usepackage{authblk}

\usepackage{xcolor}
\newcolumntype{P}[1]{>{\centering\arraybackslash}p{#1}}
\usepackage{graphicx}
\graphicspath{ {./photos/} }
\usepackage{hyperref}

\begin{document}

\newtheorem{theorem}{Theorem}
\newtheorem{lemma}[theorem]{Lemma}

\title{Robustness of Bell Violation of Graph States to Qubit Loss}

\author{Shahar Silberstein}
\author{Rotem Arnon-Friedman}
\affil{The Center for Quantum Science and Technology, Department of Physics of Complex Systems,
Weizmann Institute of Science,
Rehovot 76100, Israel}

\maketitle

\begin{abstract}
Graph states are special entangled states advantageous for many quantum technologies, including quantum error correction, multiparty quantum communication and measurement-based quantum computation. Yet, their fidelity is often disrupted by various errors, most notably qubit loss. 
In general, given an entangled state, Bell inequalities can be used to certify whether quantum entanglement remains despite errors. 
Here we study the robustness of graph states to loss in terms of their Bell violation. Treating the recently proposed linearly scalable Bell operators by Baccari \textit{et al.}, we use the stabilizer formalism to derive a formula for the extent by which the Bell violation of a given graph state is decreased with qubit loss. Our analysis allows to determine which graph topologies are tolerable to qubit loss as well as pinpointing the Achilles' heel of each graph, namely the sets of qubits whose loss jeopardizes the Bell violation. These results can serve as an analytical tool for optimizing experiments and protocols involving graph states in real-life systems.
\end{abstract}

\section{Introduction}



    Graph states \citep{Hein_2004} are special entangled states that can be characterized by a  graph in which each vertex represents a qubit. They are defined to satisfy a set of ``stabilizing constraints'' applying to each qubit and its neighbors. The states are useful for many quantum tasks and technologies, ranging from quantum error correction \citep{Schlingemann_2002, Looi_2008, Liao_2022}, to multiparty cryptographic protocols \citep{Bell_2014_2, Hein_2006}, to all-photonic quantum repeaters \citep{Azuma_2015} and computers \citep{Raussendorf_2003}.
    As they hold such great importance, a vast amount of research is devoted to investigating their theoretical aspects and creating them in laboratories using different quantum technologies \citep{Kay_2006, Zhang_2006, Lu_2007, Bell_2014_1, Adcock_2019, Pick_2021, Mooney_2019}.
    
    As graph states are entangled states, there exist Bell inequalities that can be violated by them \citep{Gisin_1992}. The violation gives a very strong form of ``certification'' of the states \citep{Supic_2018}. 
    These facts allow us to use graph states as a quantum resource in the so called \emph{device-independent} (DI) setting. 
    In the DI paradigm~\cite{ekert2014ultimate} one wishes to certify that an unknown quantum state is sufficiently good in order to accomplish a given task (e.g., the production of randomness~\cite{Colbeck09}). This is done by measuring the state using uncharacterized devices and checking whether the resulting classical information, namely, the measurement outcomes, can be used to violate a Bell inequality~\cite{brunner2014bell,scarani2019bell}. 
    
    In recent years, several Bell inequalities were developed specifically for graph states \citep{Scarani_2005, Guhne_2005, Toth_2006, Baccari_2020, Santos_2022}. The inequalities differ in the number of prescribed measurement settings, 
    the maximal violation ratio with respect to the classical bound and more. In this work we are interested in the robustness of the Bell violation of the inequalities to the \emph{loss of qubits}. A loss of qubits is the most common source of error in many setups. Most drastically, attenuation causes loss when photonic qubits are being sent over a quantum communication channel \citep{Yin_2016, Chen_2020, Czerwinski_2022}. 
    
    Another reason to consider loss comes from the field of multipartite computation and cryptography \citep{Thalacker_2021}. Consider a setup in which a graph state is being used as a resource for some task where, say, each entity holds one qubit of the graph state. 
    If one of the parties decides to stop participating in the protocol, this results, effectively, in the loss of a qubit. Ideally, one would like the protocol to be robust to this loss so the task can be accomplished regardless.
    
    Some studies regarding the effects of loss of qubits on graph states in the context of error correction have been performed, e.g., \citep{Stace_2009_a, Stace_2009_b, Bell_2022}. We stress that our work considers the \emph{DI} setup, where the initial state, the set of possible measurements and the errors are unknown. Hence, the more demanding Bell violation is asked for, as apposed to having an error model and an error correcting code that can be used to handle qubit losses. \citep{Noller_2023} discussed the implication for specific types of hyper-graphs. Here we provide a general theorem for the violation of a lineally scalable Bell inequality \citep{Baccari_2020} which can be applied to any graph state. Apart from the fundamental aspects, our findings can be used to construct loss-tolerant DI protocols that use graph states. 
    
    Our main goal is to understand how the violation of the Bell inequality decreases when some qubits of a given graph state are lost (see Section~\ref{results_section}). In addition, we show that for any graph state, when a qubit from the set of roots is lost there is no violation. 
    We present several examples for \emph{loss-tolerant} (Section~\ref{loss_tolerant_section}) and \emph{loss-sensitive} graph states (Section~\ref{loss_sensitive_section}) and find that a redundancy of roots is beneficial to construct loss-tolerant graph states.
    In each example, we indicate which qubits are ``more important'' than others and how many qubits can be lost while still allowing to observe a violation of the Bell inequalities (see Table~\ref{table:conditions_for_violation} for a summary). Lastly, we explore the implication of the possibility of the set of lost qubits being unknown (Section~\ref{unkown_set_of_lost_qubits_section}).

    On top of the theoretical contribution, our results and examples are of relevance for experiments. 
    Firstly, if loss is to be expected, one should clearly aim to create loss-tolerant graph states; our results show how many qubits can be lost for a given state while still presenting a Bell violation.  
    Then, since some qubits are crucial for the violation, it is  wise to treat them differently in an experiment. For example, have them as, e.g., atomic qubits or be measured right away, while others can be sent as photons.
    Thus, our work lays out important facts for consideration when developing experiments that produce graph states.

\section{Preliminaries}
\subsection{Bell Scenarios}
We consider \(\left(N,2,2\right)\) Bell Scenarios, which can be thought of as “cooperative games” with~$N$ participating parties. Each participant \(i\) receives an input \(x_{i}\in\mathcal{I}_{i}=\left\{ 0,1\right\}\) sampled from some probability distribution \(q\left(\left\{ x_{i}\right\}_{i=1}^{N}\right)\), and produces an output \(y_{i}\in\mathcal{O}_{i}=\left\{ -1,1\right\}\). The score of the game is defined via a function \(v\left(\left\{ x_{i}\right\} _{i=1}^{N},\left\{ y_{i}\right\} _{i=1}^{N}\right)\).
The behavior of the parties can be described by some conditional probability distribution \(P\left(\left\{y_{i}\right\}_{i=1}^{N} \Big| \left\{x_{i}\right\}_{i=1}^{N}\right)\), which is also called the \textit{strategy} of the game.

The score of a strategy is given by 

\begin{align}
    I\left(P\right)	=\sum_{\left\{ x_{i}\right\} _{i=1}^{N},\left\{ y_{i}\right\} _{i=1}^{N}}q\left(\left\{ x_{i}\right\} _{i=1}^{N}\right) P\left(\left\{y_{i}\right\}_{i=1}^{N} \Big| \left\{x_{i}\right\}_{i=1}^{N}\right) v\left(\left\{ x_{i}\right\} _{i=1}^{N},\left\{ y_{i}\right\} _{i=1}^{N}\right)
	\;.
\end{align}

A strategy is called \textit{local}, or classical, if there exists a family of local processes \(\forall \lambda : \left\{ P_{1}^{\lambda}\left(y\mid x\right),\dots,P_{N}^{\lambda}\left(y\mid x\right)\right\}\) and a distribution \(Q\left(\lambda\right)\) such that 
	
\begin{align}
	P\left(\left\{y_{i}\right\}_{i=1}^{N} \Big| \left\{x_{i}\right\}_{i=1}^{N}\right) = \int d\lambda Q\left(\lambda\right)\Pi_{i=1}^{N}P_{i}^{\lambda}\left(y_{i}\mid x_{i}\right) 
	\;.
\end{align}
A strategy is called \textit{deterministic} if the set of outputs is predetermined for any set of inputs, i.e
	\begin{align}
		P\left(\left\{y_{i}\right\}_{i=1}^{N} \Big| \left\{x_{i}\right\}_{i=1}^{N}\right) = \delta_{\left(y_{1},\dots,y_{n} \right) = f\left(x_{1},\dots,x_{n}\right)}
		\;,
	\end{align}
where \(f\) is some function of the inputs. Hence, a strategy is said to be \textit{local deterministic} if it relies on a predetermined agreement of the outputs to each input separately. According to Fine's theorem \citep{Fine_1982}, any local strategy can be written as a convex combination of local deterministic strategies. The classical bound \(\beta_{\text{C}}\) of the game is defined as the maximal average score a local strategy can achieve. 

In Bell scenarios, the strategies which involve \textit{quantum resources} can achieve better scores. A strategy which uses a quantum resource in this context is described by a quantum density matrix \(\rho\) acting on  \(\bigotimes_{i=1}^{N}\mathcal{H}_{i}\) where \(\left\{ \mathcal{H}_{i}\right\} _{i=1}^{N}\) are Hilbert spaces, and families of measurements
\begin{align}
 \forall i,\forall x_{i}\in\mathcal{I}_{i}:\left\{ \Pi_{i}\left(x_{i}, y_{i}\right)|y_{i}\in\mathcal{O}_{i}\right\}   ,\sum_{y_{i}\in\mathcal{O}_{i}}\Pi_{i}\left(x_{i} ,y_{i}\right)=\mathbb{I}_{i} 
 \;,
\end{align}
such that
\begin{align}
	P\left(\left\{y_{i}\right\}_{i=1}^{N} \Big| \left\{x_{i}\right\}_{i=1}^{N}\right)=\text{Tr}\left(\rho\bigotimes_{i=1}^{N}\Pi_{i}\left(x_{i},y_{i}\right)\right)
	\;.
\end{align}

The quantum bound of the game \(\beta_{\text{Q}}\) is the maximal average score a quantum strategy can achieve. In a Bell scenario \(\beta_{\text{Q}}\geq\beta_{\text{C}}\).

A correlator of the outputs is a function of the form 
\begin{align}
    \label{correlator_of_inputs}
    \left\langle y_{i_{1}}\cdot\cdot\cdot y_{i_{k}} \right\rangle \left(\left\{ x_{i}\right\} _{i=1}^{N} \right)= 
    \sum_{\left\{ y_{i}\right\}_{i=1}^{N}} y_{i_{1}}\cdot\cdot\cdot y_{i_{k}} P\left(\left\{y_{i}\right\}_{i=1}^{N} \Big| \left\{x_{i}\right\}_{i=1}^{N}\right)
    \;,
\end{align}
where \(\left\{ i_{1},\dots,i_{k}\right\} \subseteq\left\{ 1,\dots,N\right\}\). In many cases, \(I\left(P\right)\) is a linear combination of such correlators. For a quantum strategy Equation~\eqref{correlator_of_inputs} can be written as an operator expectation value

\begin{align}
    \left\langle y_{i_{1}}\cdot\cdot\cdot y_{i_{k}} \right\rangle \left(\left\{ x_{i}\right\} _{i=1}^{N} \right) = \left\langle  Y_{i_{1}}\left(x_{i_{1}}\right)\dots Y_{i_{k}}\left(x_{i_{k}}\right) \right\rangle
    \;,
\end{align}
where \(Y_{i} \left(x_{i}\right) =  \sum_{y_{i} \in\mathcal {O}_{i}}y_{i}\cdot\Pi_{i}\left(x_{i},y_{i}\right)\).

\subsection{Graph States}\label{sec:pre_graph}
The \(\left(N,2,2\right)\) Bell inequalities examined in this paper are maximally violated by N-qubit states called graph states, which are defined as follows.
Let \(G=\left(V,E\right)\) be a graph, where \(V=\left\{v_{1},\dots,v_{N}\right\}\) is the set of vertices (\(\left|V\right|=N\)) and 
\(E\) is the set of edges. Each of the qubits in the system is associated with a vertex in \(V\). Denote by \(\mathcal{N}_{i}^{G}\) the neighborhood of the vertex \(v_{i}\), i.e., \(\mathcal{N}_{i}^{G} = \left\{j: \left(v_{i},v_{j}\right)\in E \right\}\) and define \(\overline{\mathcal{N}_{i}^{G}} \equiv  \mathcal{N}_{i}^{G} \cup \left\{i\right\}\). To every vertex \(v_{i}\in V\) associate a stabilizing operator
	
\begin{align}
	\label{eqn:graph_stablizing_operators}
	S_{i}^{G} \equiv X_{i} \bigotimes_{j\in \mathcal{N}_{i}^{G}} Z_{j}
	\;,
\end{align}
where the identity operator acts on all sites not appearing explicitly. Note, that the operators \(\left\{S_{i}^{G}\right\}^{N}_{i=1}\) commute with one another and therefore can be diagonalized simultaneously. The graph state \(\left|\phi^{G}\right\rangle\) associated with the graph \(G\) is an N-qubit state defined as the unique\footnote{Generally, for a stabilizer code with \(n\) qubits and \(k\) stabilizers, the corresponding code space has \(2^{n-k}\) states. This is proved by induction by adding one stabilizer at a time to the stabilizer group and showing that with each addition the number of eigenstates with eigenvalue \(1\) for all stabilizers is cut in half. For more information see \citep{Gottesman_1997}.} eigenstate of all the stabilizers with eigenvalue \(1\), i.e
\begin{align}	
	\label{eqn:eigenvalue_equation}
	\forall i\in \left\{1,\dots,N\right\} :  S_{i}^{G}\left|\phi^{G}\right\rangle = \left|\phi^{G}\right\rangle
	\;.
\end{align}

A node with the maximal number of neighbors \(n_{\text{max}}\) is called a \textit{root}. Denote without loss of generality one of the roots by the index~\(r\).
Define the following Bell operator maximally violated by graph states, proposed in  \citep{Baccari_2020} 

\begin{align}
\label{eqn:bell_operator_sccarni}
    I_{r}^{G}  \left[\left\{Y_{i}\left(0\right)\right\}_{i=1}^{N},\left\{Y_{i}\left(1\right)\right\}_{i=1}^{N}\right] &= n_{\text{max}}\left[Y_{r}\left(0\right) + Y_{r}\left(1\right)\right]
	\bigotimes_{i\in\mathcal{N}_{r}^{G}} Y_{i}\left(1\right) \nonumber\\
	& + \left[Y_{r}\left(0\right) - Y_{r}\left(1\right)\right]\sum_{i\in\mathcal{N}_{r}^{G}}
	Y_{i}\left(0\right)\bigotimes_{j\in\mathcal{N}_{i}^{G}\setminus\left\{r\right\}} Y_{j}\left(1\right) \nonumber\\
	& + \sum_{i\notin\overline{\mathcal{N}_{r}^{G}}}
	Y_{i}\left(0\right)\bigotimes_{j\in\mathcal{N}_{i}^{G}} 
	Y_{j}\left(1\right)
	\;.
\end{align}
An example of a graph with two possible root vertices and \(n_{\text{max}} = 4\) is shown in Figure~\ref{fig:graph_example}.

\begin{figure}[h]
	\centering
	\begin{subfigure}[b]{6cm}
    	\centering
    	\includegraphics[height=3cm]{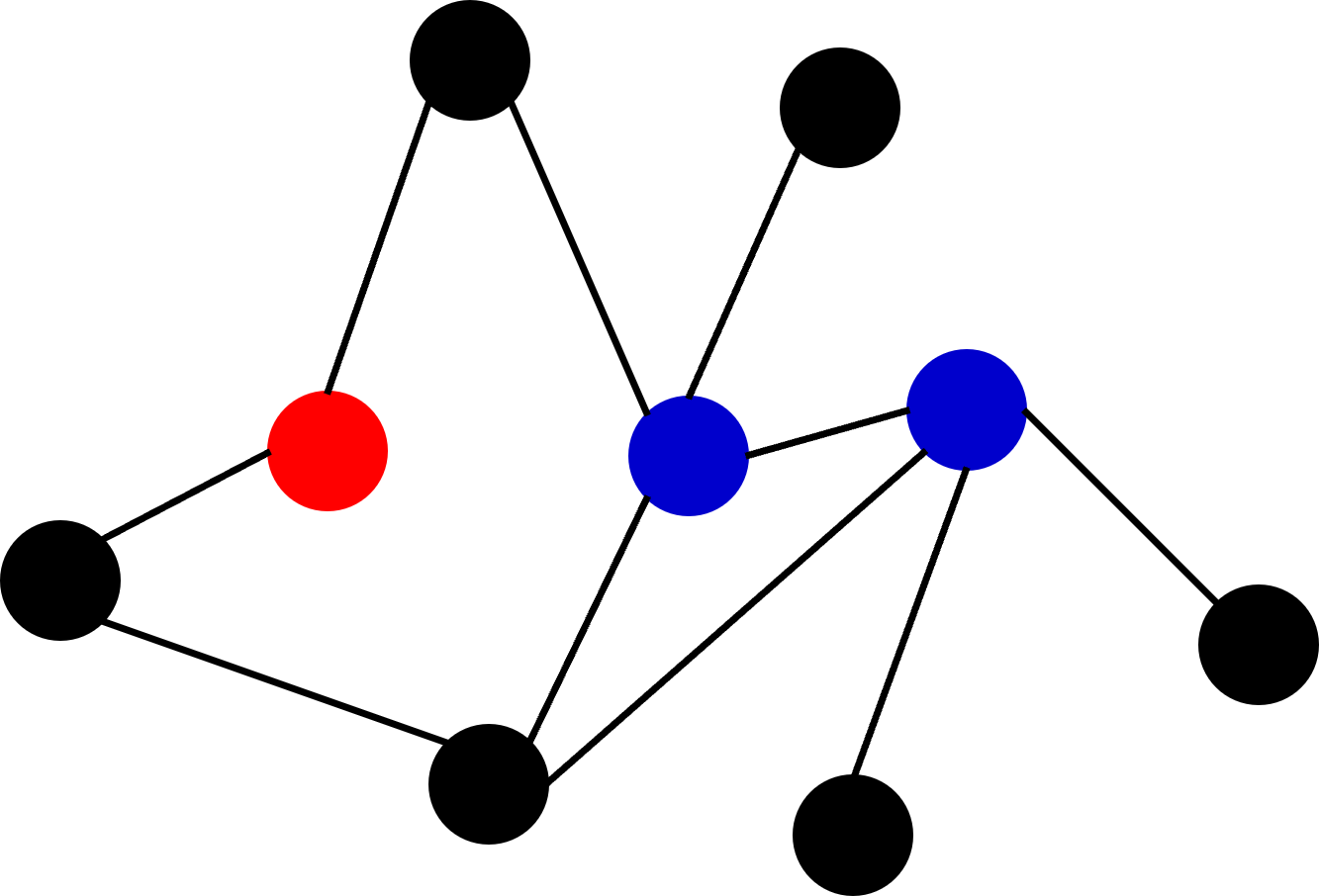}
    	\caption{}
    	\label{fig:graph_example}
	\end{subfigure}
	\quad
	\begin{subfigure}[b]{6cm}
	    \centering
		\includegraphics[height=3cm]{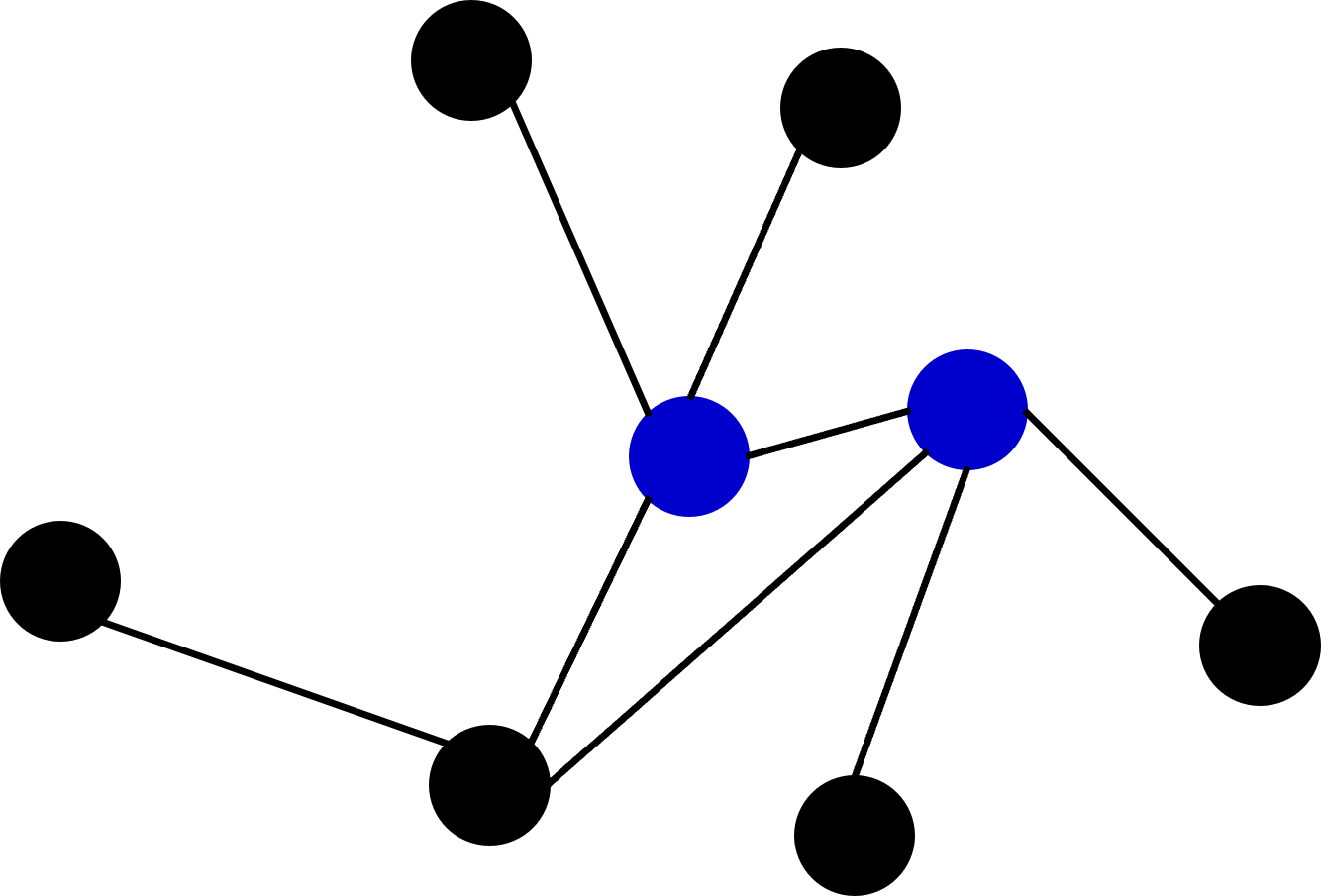}
		\caption{}
		\label{fig:induced_graph}
	\end{subfigure}
	\caption{(\ref{fig:graph_example})~A graph with two vertices with the maximal number of neighbors, which are denoted in blue. (\ref{fig:induced_graph})~The induced graph after eliminating the vertex colored in red in Figure \ref{fig:graph_example}.}
\end{figure}

It is shown that this Bell operator has the following classical bound 

\begin{align}
	\beta_{\text{C}}=n_{\text{max}}+N-1
	\;,
\end{align}
and the maximal quantum violation is
\begin{align}
\label{quantum_bound}
	\beta_{\text{Q}}=\left(2\sqrt{2}-1\right)n_{\text{max}}+N-1
	\;.
\end{align}
 The maximal expectation value in Equation~\eqref{quantum_bound} can be achieved by using the following measurements and the graph state
\begin{align}
\label{eqn:measurments}
	Y_{i}\left(0\right) = 
	\begin{cases}
		\frac{1}{\sqrt{2}}\left(X_{i}+Z_{i}\right) & i = r \\
		X_{i} & \text{otherwise}
	\end{cases} && 
	Y_{i}\left(1\right) =
	\begin{cases}
		\frac{1}{\sqrt{2}}\left(X_{i}-Z_{i}\right) & i = r \\
		Z_{i} & \text{otherwise}
	\end{cases}
	\;.
\end{align}

In this case\footnote{The * in \(I_{r}^{G*}\) in Equation~\eqref{eqn:bell_operator_as_a_sum_of_stabilizing_operators} stands for the specific choice of measurements described in Equation~\eqref{eqn:measurments}}
\begin{align}
	\label{eqn:bell_operator_as_a_sum_of_stabilizing_operators}
	I_{r}^{G*}  &= \sqrt{2}n_{\text{max}}X_{r}
	\bigotimes_{i\in\mathcal{N}_{r}^{G}} Z_{i} \nonumber\\
	& + \sqrt{2}Z_{r} \sum_{i\in\mathcal{N}_{r}^{G}}
	X_{i}
	\bigotimes_{j\in\mathcal{N}_{i}^{G}\setminus\left\{r\right\}} Z_{j} \nonumber\\
	& + \sum_{i\notin \overline{\mathcal{N}_{r}^{G}}}
	X_{i} \bigotimes_{j\in\mathcal{N}_{i}^{G}} 
	Z_{j} \nonumber\\ 
	& = 
	 \sqrt{2}n_{\text{max}}S_{r}^{G}
	+ \sqrt{2}\sum_{i\in\mathcal{N}_{r}^{G}} S_{i}^{G} + 
	\sum_{i\notin\overline{\mathcal{N}_{r}^{G}}} S_{i}^{G}
	\;.
\end{align}

Notice that the above is a sum of the stabilizing operators associated to each vertex of the graph as in Equation~\eqref{eqn:graph_stablizing_operators}. For each one of the operators the corresponding graph state has eigenvalue \(1\) (This is part of the reason for which this Bell operator is maximally violated by graph states). Other examples of Bell inequalities violated by graph states were shown in \citep{Scarani_2005}.

Given a graph \(G=\left(V,E\right)\), let \(V' \subset V\) be any subset of vertices of \(G\). The \textit{induced subgraph} \(G\left[V'\right]\) is the graph whose vertex set is \(V'\) and whose edge set consists of all of the edges in \(E\) that have both endpoints in \(V'\). An example of an induced graph of the graph in Figure~\ref{fig:graph_example} is shown in Figure~\ref{fig:induced_graph}, where an excluded vertex is marked in red. In this paper we discuss the effect of loss of qubits on the violation of the Bell operators. When qubits are lost, it is natural to examine the violation of the Bell operator corresponding to the original graph or to the induced subgraph defined above. 

\section{Results}\label{sec:results}

We now consider the effect of loss of qubits on the possibility of violating the Bell inequalities presented in Section~\ref{sec:pre_graph}. The loss-tolerance is determined by the decrease of the violation when losing qubits of the ideal graph state and while using the ideal measurements.

\subsection{Graph States' Violation of Bell Inequalities after Qubit Loss}
\label{results_section}

The quantitative effect of tracing out qubits on the expectation value of the operator in Equation~\eqref{eqn:bell_operator_as_a_sum_of_stabilizing_operators}, for both the original graph and the induced graph, is the content of Theorem~\ref{expectation_value_theorem}.

\begin{theorem}
\label{expectation_value_theorem}
    Let \(\left|\phi^{G}\right\rangle\) be a graph state defined on the graph \(G=\left(V,E\right)\), (\(\left|V\right|=N\)). Let \(\mathcal{L}\) be a set of indices \(\mathcal{L} \subseteq \left\{ 1,\dots,N \right\}\) of the lost qubits. Define the density matrix after qubit loss\footnote{The choice of representing the lost qubits by replacing them with \(\bigotimes_{l \in \mathcal{L}} \left|0\right\rangle_{l}\left\langle0\right|_{l}\) is arbitrary, the same results would have been obtained by replacing the lost qubit set with \(\bigotimes_{l \in \mathcal{L}} \left|1\right\rangle_{l}\left\langle1\right|_{l}\) or the totally mixed state.}
     \begin{align}
	    \label{eqn:remaining_density_matrix}
		\rho = \mathrm{Tr}_{\mathcal{L}}\left( \left|\phi^{G}\right\rangle \left\langle \phi^{G}\right| \right) \bigotimes_{l \in \mathcal{L}} \left|0\right\rangle_{l}\left\langle0\right|_{l} \;.
  \end{align}
  Denote by  \(n_{\text{max}}\) the maximal number of neighbors of a vertex in the graph and by \(r\) an index of a vertex with the maximal number of neighbors such that \(r\notin\mathcal{L}\).\footnote{For \(r\in\mathcal{L}\) the equation holds for \(I_{r}^{G*}\).} 
  
  Define the sets \(\mathcal{W}^{G} = \mathcal{N}_{r}^{G} \setminus  \bigcup_{l \in \mathcal{L}} \overline{\mathcal{N}_{l}^{G}} \) and \( \mathcal{T}^{G} = \left\{1,\dots,N\right\} \setminus \left( \bigcup_{l \in \mathcal{L}} \overline{\mathcal{N}_{l}^{G}} \cup \overline{\mathcal{N}_{r}^{G}}\right)\).
	Then,
	\begin{align}
		\label{Theorem_content}
		\left\langle I_{r}^{G*} \right\rangle_{\rho} = \left\langle I_{r}^{G\left[V\setminus \left\{ v_{l} \text{ s.t } l\in \mathcal{L}\right\}\right]*} \right\rangle_{\rho} = 
		\begin{cases}
			\sqrt{2}n_{\text{max}}+\sqrt{2}\left|\mathcal{W}^{G}\right|+\left|\mathcal{T}^{G}\right| &  \overline{\mathcal{N}_{r}^{G}}  \cap \mathcal{L} = \emptyset \\
			\sqrt{2}\left|\mathcal{W}^{G}\right|+\left|\mathcal{T}^{G}\right| &  \text{otherwise}
		\end{cases}
		\;,
	\end{align}
	
	where \(I_{r}^{G*}, I_{r}^{G\left[V\setminus \left\{ v_{l} \text{ s.t } l\in \mathcal{L}\right\}\right]*}\) are defined via Equation~\eqref{eqn:bell_operator_as_a_sum_of_stabilizing_operators}.
\end{theorem}

An example of a graph along with the sets \(\mathcal{W}\) and \(\mathcal{T}\) defined in Theorem~\ref{expectation_value_theorem} is shown in Figure~\ref{fig:theorem_example}.

\begin{figure}
	\centering
	\includegraphics[height=3cm]{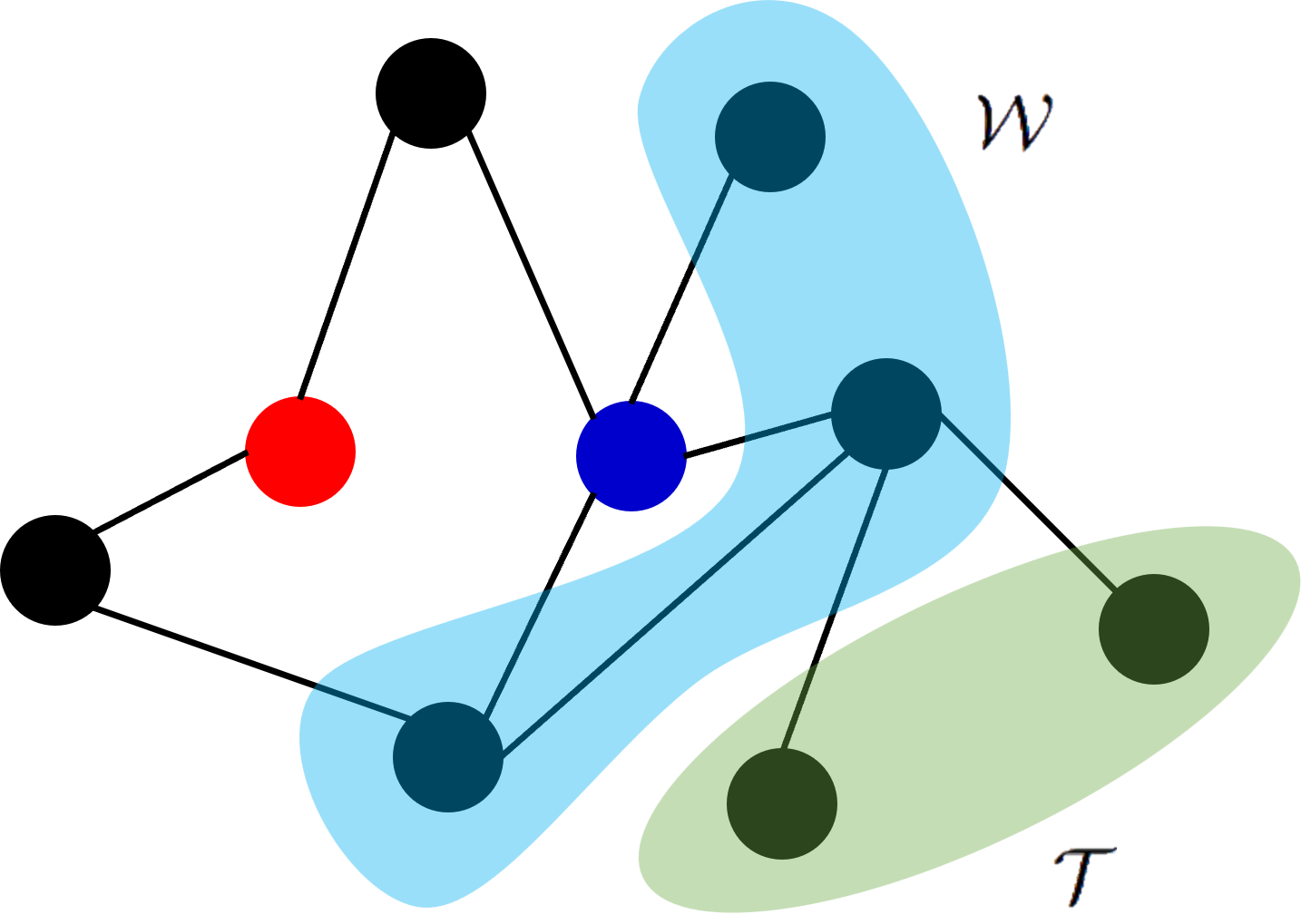}
	\caption{An example of the sets contributing to the expectation values \(I_{r}^{G*}, I_{r}^{G\left[V\setminus \left\{ v_{l} \text{ s.t } l\in \mathcal{L}\right\}\right]*}\) appear in Figure~\ref{fig:theorem_example}. A root vertex is denoted in blue, and a lost qubit is marked in red. The sets \(\mathcal{W}, \mathcal{T}\) defined in Theorem~\ref{expectation_value_theorem} which are not neighbors of the lost qubit set and have a different contribution to the expectation value are marked in light blue and light green respectively.}
	\label{fig:theorem_example}
\end{figure}

The Bell operators maximally violated by the corresponding graph states described in Equation~\eqref{eqn:measurments} and Equation~\eqref{eqn:bell_operator_as_a_sum_of_stabilizing_operators} are a sum of stabilizing operators defined in Equations~\eqref{eqn:graph_stablizing_operators}. The following lemma describes the effect of tracing out qubits on the expectation value of each stabilizer; The proof of Theorem~\ref{expectation_value_theorem} is a direct result of it. 

\begin{restatable}{lemma}{universe}
\label{lost_qubits_lemma}
Let \(G, \mathcal{L}, \rho\) be as in Theorem~\ref{expectation_value_theorem}.
Let \(S_{i}^{G}\) be the stabilizing operator associated with the vertex \(v_{i} \in V\) corresponding to the graph \(G\). Then,
	\begin{enumerate}[label=(\alph*)]
	\item 
    	\begin{align}
    	\nonumber
    	    \left\langle S_{i}^{G} \right\rangle_{\rho}  = 
    	    \begin{cases}
    		   0 & i \in \bigcup_{l \in \mathcal{L}} \overline{\mathcal{N}_{l}^{G}} \\
    		1 & \text{otherwise}
    	    \end{cases}
    	\end{align}
    \item
        \begin{align}
    	\label{eqn:stablizer_expectation_value}
        	    \left\langle S_{i}^{G\left[V\setminus \right\{ v_{k} \text{ s.t. } k\in\mathcal{L} \left\}\right]} \right\rangle_{\rho}  = 
        	    \begin{cases}
        		   0 & i \in \bigcup_{l \in \mathcal{L}} \mathcal{N}_{l}^{G} \\
        		1 & \text{otherwise}
        	    \end{cases}
        	     \;.
        \end{align}
	\end{enumerate}
\end{restatable}

The full proof of Lemma~\ref{lost_qubits_lemma} is shown in Appendix~\ref{appendix:full_lemma_proof}, with the main ideas mentioned here. It becomes apparent that after the partial trace, the remaining state is \emph{not} a graph state of the induced graph. Stabilizing operators corresponding to neighbors of traced out qubits \(i \in \left(\bigcup_{l \in \mathcal{L}} \mathcal{N}_{l}^{G}\right)\) do not have expectation value \(1\) in the new state.

The classical limits of the Bell operators corresponding to the graphs \(G\) and \( G\left[V\setminus \left\{ v_{l} \text{ s.t } l\in \mathcal{L}\right\}\right]\) generally satisfy \(\beta^{G\left[V\setminus \left\{ v_{l} \text{ s.t } l\in \mathcal{L}\right\}\right]}_{\text{C}} < \beta^{G}_{\text{C}}\), since the number of vertices is smaller in the induced graph when \(\mathcal{L} \neq \emptyset\).
As shown in Theorem~\ref{expectation_value_theorem}, the expectation values of the operators \(I_{r}^{G*}\) and \( I_{r}^{G\left[V\setminus \left\{ v_{l} \text{ s.t } l\in \mathcal{L}\right\}\right]*}\) with respect to the density matrix in Equation~\eqref{eqn:remaining_density_matrix} are the same. Hence, it is more likely to observe violation of the operator \(I_{r}^{G\left[V\setminus \left\{ v_{l} \text{ s.t } l\in \mathcal{L}\right\}\right]*}\). Hence, one can gather the statistics performing measurements with respect to the original full graph and then compare the expectation value to the classical limit of the induced graph.

When attempting to use graph states in scenarios in which qubit loss is dominant, the following theorem shows that for every graph, a loss of a root jeopardizes the violation.

\begin{theorem}
\label{all_roots_ruin_violation}
    Let \(G, \mathcal{L}, \rho\) be as in Theorem~\ref{expectation_value_theorem} and let \(\mathcal{R}\) be the set of all roots of \(G\). Denote by \(\mathcal{B}=\left\{ I_{r}^{G*} | \; r\in\mathcal{R}\right\}\) the set of Bell operators defined in Equation~\eqref{eqn:bell_operator_as_a_sum_of_stabilizing_operators} with the specific measurement bases defined in Equation~\eqref{eqn:measurments}. Then, \(\forall r' \in \mathcal{R}\) if \(\ \mathcal{L} = \left\{r'\right\}\), it follows that\footnote{In Equation \eqref{eqn:all_roots_ruin_violation}, \(\left\langle I_{r}^{G\left[V \setminus \left\{r'\right\}\right]*} \right\rangle_{\rho}\) is not always defined, it assumes \(r\) is also a root for the induced graph.} \(\forall  I_{r}^{G*} \in \mathcal{B}\)
    \begin{align}
    \label{eqn:all_roots_ruin_violation}
        \left\langle I_{r}^{G*} \right\rangle_{\rho} = \left\langle I_{r}^{G\left[V \setminus \left\{r'\right\}\right]*} \right\rangle_{\rho} < \beta_{\text{C}}^{G\left[V \setminus \left\{r'\right\}\right]} < \beta_{\text{C}}^{G}
        \;. 
    \end{align} 
\end{theorem}

\begin{proof}
    Let \(r,r'\) be as in the theorem conditions. First, assume that \(n_\text{max}\) of the induced graph remains unchanged so that \(  \beta_{\text{C}}^{G\left[V \setminus \left\{r'\right\}\right]} = N + n_{\text{max}} -2\) and that \(r \neq r'\). If some root is lost,  \(n_{\text{max}} + 1\) stabilizers in Equation~\eqref{eqn:bell_operator_as_a_sum_of_stabilizing_operators} vanish according to Lemma~\ref{lost_qubits_lemma} (corresponding to the lost qubit and its neighbors). The coefficients of all stabilizers in Equation~\eqref{eqn:bell_operator_as_a_sum_of_stabilizing_operators} are \(1, \sqrt{2}\) or \(\sqrt{2}n_{\text{max}}\), meaning that the quantum bound is decreased at
     least by \(n_{\text{max}} + 1\). Hence, \(\forall  I_{r}^{G*} \in \mathcal{B}\)
     \begin{align}
         \left\langle  I_{r}^{G*} \right\rangle \leq \beta_{\text{Q}}^{G} - \left(n_{\text{max}} + 1\right) = \left(2\sqrt{2}-2\right)n_{\text{max}}+N-2 < \beta_{\text{C}}^{G\left[V \setminus \left\{r'\right\}\right]}
     \;.
     \end{align}
     If \(n_\text{max}\) is decreased, it follows that \(r' \in \mathcal{N}_{r}^{G}\) which means that the expectation value is decreased by at least \(\left(\sqrt{2} + 1\right)n_{\text{max}} + \sqrt{2}  - 1\), meaning that there is a violation even though the classical limit is decreased by a larger amount.
      \(\beta_{\text{C}}^{G\left[V \setminus \left\{r'\right\}\right]} < \beta_{\text{C}}^{G}\) as mentioned before so this holds for \( \beta_{\text{C}}^{G} \) as well. The cases where \(r = r'\) or when there is a single root are treated similarly.
\end{proof}

Theorem~\ref{all_roots_ruin_violation} implies that when a root vertex is lost there is no violation of the classical limit and hence for every choice of graph state, the set of roots must be kept safe when attempting to use the Bell operator proposed in Equation~\eqref{eqn:bell_operator_sccarni} as the operator for a Bell test. The margin by which the Bell operator is not violated is generally expected to decrease further with additional lost qubits. In some extreme examples this does not hold (such as graphs where roots belong to clusters which are almost disconnected and a whole cluster of qubits is lost).

\subsection{Examples}
In this section, we discuss several applications of Theorem~\ref{expectation_value_theorem}. Different examples of graph states are examined, with the objective of identifying structures which are more likely to violate the classical limit, even when some qubits are lost. We start with well-known examples which are found to be loss-sensitive, move to loss-tolerant examples and in the end discuss the case where the set of lost qubits is unknown. In order to simplify the notation, induced graphs of the form \(G\left[V\setminus \left\{ v_{l} \text{ s.t } l\in \mathcal{L}\right\}\right]\) are denoted by \(G'\). In the entire section, \(N\) stands for the total number of vertices in the graph, as before.

\subsubsection{Loss-Sensitive Graph States}
\label{loss_sensitive_section}
\paragraph{Ring graph state.} An example of a ring graph state is shown in Figure~\ref{fig:ring}. The classical limit is given by \(\beta_{\text{C}}^{\text{Ring}}=N+1\) and is reduced by \(1\) with every lost qubit (as long as the maximal number of neighbors remains \(2\)). In this state, all the vertices are roots, and therefore as an immediate result of Theorem~\ref{all_roots_ruin_violation} there is no violation when any qubit is lost.   

\begin{figure}
	\centering
	\begin{subfigure}[b]{6cm}
	\centering
		\includegraphics[height=3cm]{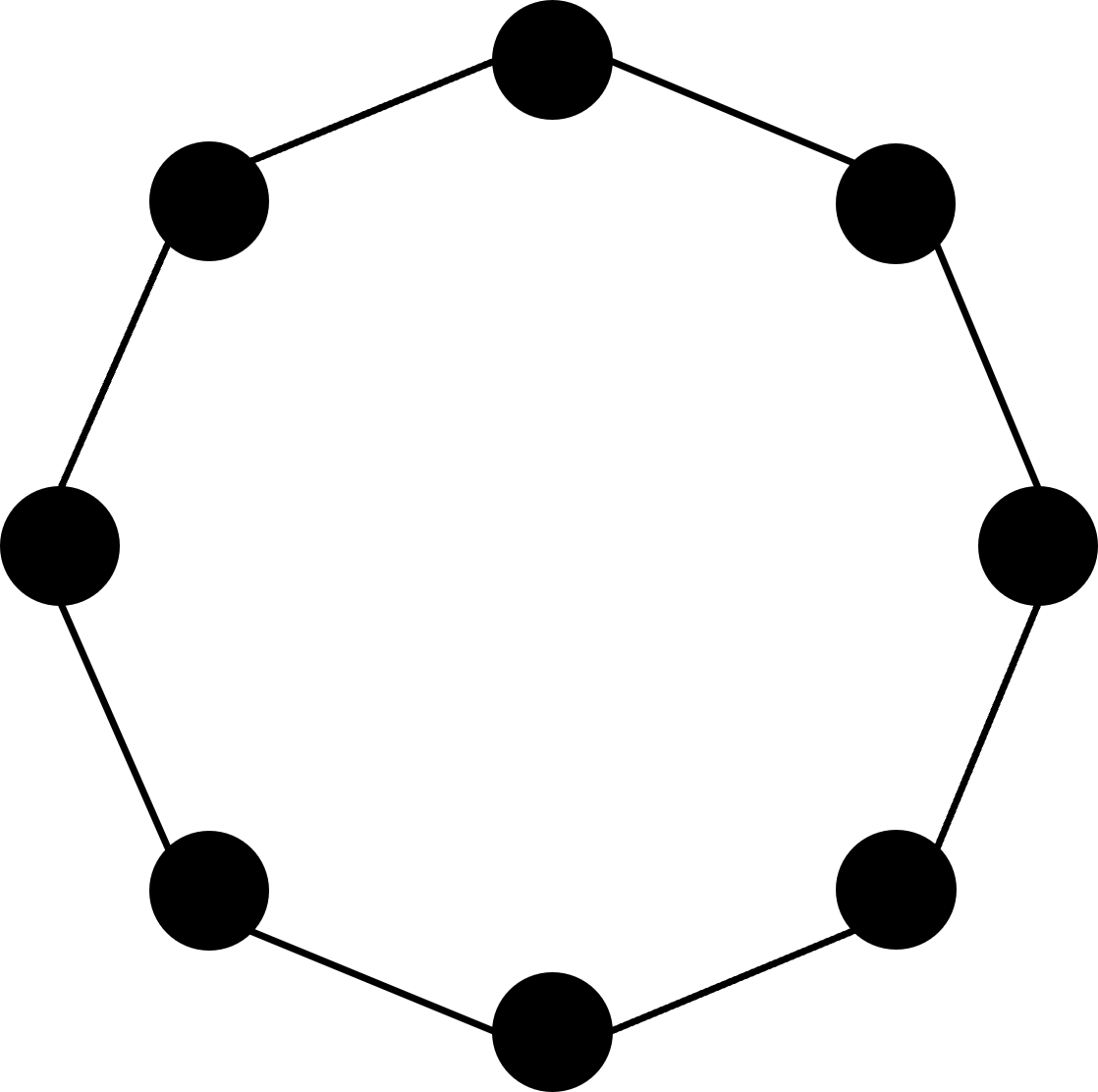}
		\caption{}
		\label{fig:ring}
	\end{subfigure}
	\quad
	\begin{subfigure}[b]{6cm}
	    \centering
		\includegraphics[height=3cm]{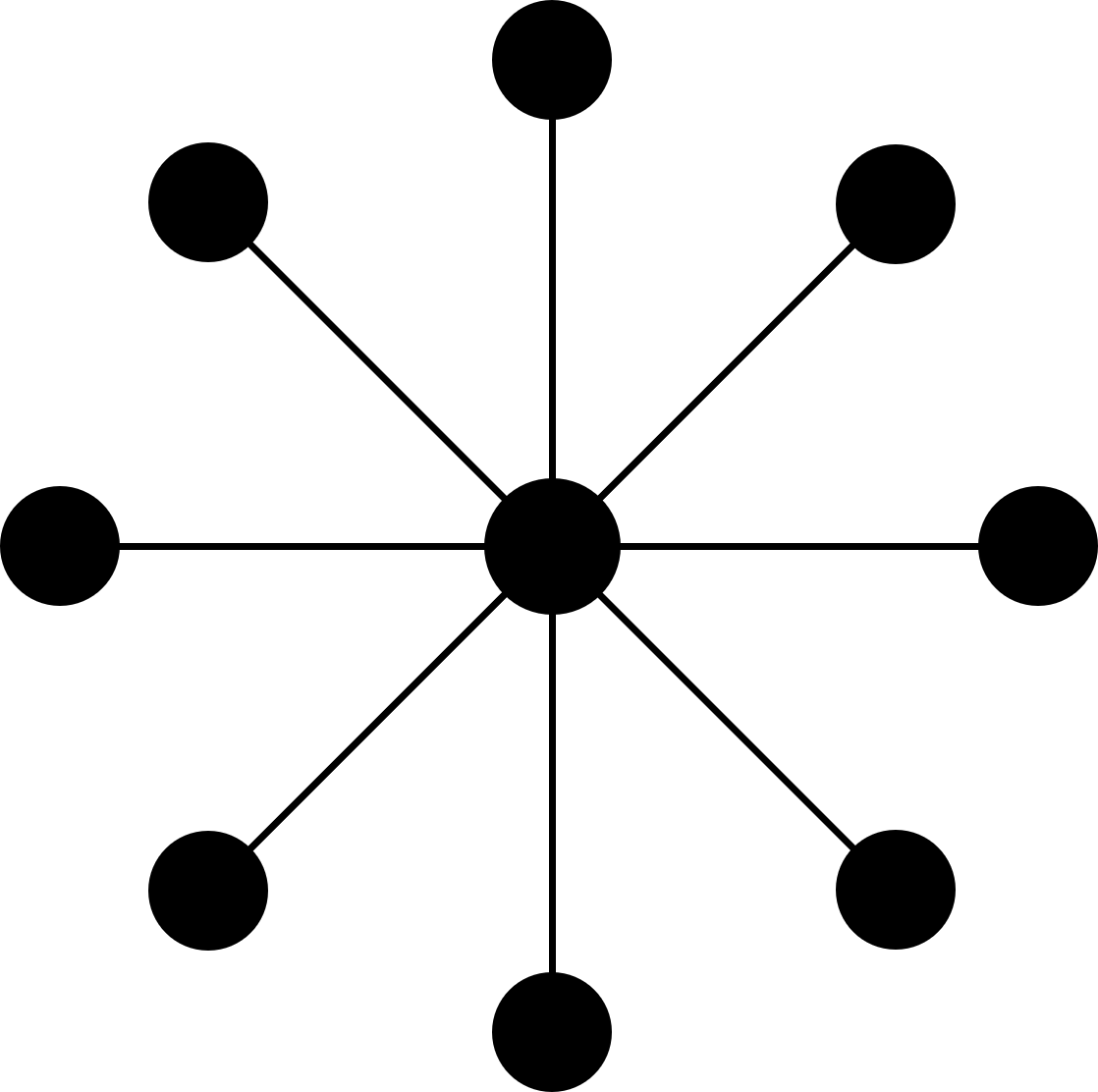}
		\caption{}
		\label{fig:ghz}
	\end{subfigure}
	\caption{Loss sensitive graph examples. For both examples there is no violation of the classical limit. (\subref{fig:ring}) A ring graph.  (\subref{fig:ghz}) A GHZ (``star'') graph. }
\end{figure}

\paragraph{GHZ (``star'') graph state.} The GHZ state is the graph state corresponding to a ``star'' graph (Figure~\ref{fig:ghz}) up to local unitaries. The classical limit is given by \(\beta_{\text{C}}^{\text{GHZ}}=2N - 2\). If the qubit corresponding to the center vertex is lost, the expectation value of the Bell operator vanishes as shown in the proof of Theorem~\ref{all_roots_ruin_violation}. If some other qubit is lost, the expectation value becomes
\begin{align}
    \left\langle I^{\text{GHZ*}}_{r} \right\rangle_{\rho} = \left\langle I^{\text{GHZ'*}}_{r} \right\rangle_{\rho} = \sqrt{2}\left(N-2\right) < \beta_{\text{C}}^{\text{GHZ'}}, \beta_{\text{C}}^{\text{GHZ}}
    \;.
\end{align}

The GHZ example demonstrates that when neighbors of the vertex chosen as the root are lost, a major contribution to the Bell operator's expectation value vanishes. This leads to the understanding that in order to construct graph states that are robust to loss of qubits, it can be beneficial to choose a graph with a redundancy of roots.  

In order to maintain some violation of the classical limit it is sufficient to have one choice of root succeed. A graph with only one root is sensitive to loss of any neighbor of the root or the root itself and thus having more than one root might be helpful. On the other hand, if all qubits are roots according to Theorem~\ref{all_roots_ruin_violation} there is no violation. Hence, more promising architectures must include several ``clusters'' of roots and their neighbors. This key insight leads to the construction of following examples.

\subsubsection{Loss-Tolerant Graph States}
\label{loss_tolerant_section}
\paragraph{``Two centered GHZ'' graph state.} Consider the ``two centered GHZ'' as in Figure~\ref{fig:two_centered_ghz}. The graph has two root vertices connected to one another and several leaf vertices adjacent to each one of the centers. Denote the centers by \(r_{1},r_{2}\). The classical limit in this case is given by \(\beta_{\text{C}}^{\text{Two centered GHZ}} = \frac{3}{2}N - 1\). If \(r_{1}\) or \(r_{2}\) are lost there is no violation as implied from Theorem~\ref{all_roots_ruin_violation}. 

On the other hand, suppose that a qubit corresponding to \(v_{l}\), \(l \neq r_{1},r_{2}\) is lost. Without loss of generality \(l\in\mathcal{N}\left(r_{1}\right) \). The expectation value of the Bell operator corresponding to \(r_{2}\) is given by 
\begin{align}
	\label{eqn:two_centered_star_with_lost_qubit_expectation_value}
	\left\langle I^{\text{Two centered GHZ*}}_{r_{2}} \right\rangle_{\rho} =	\left\langle I^{\text{Two centered GHZ'*}}_{r_{2}} \right\rangle_{\rho} = \left(\sqrt{2} + \frac{1}{2}\right)N -\sqrt{2} - 2
	 \;.
\end{align}
Hence the classical limit is violated.
If qubits adjacent to both \(r_{1}, r_{2}\) are lost, there is no violation. However, if multiple qubits adjacent to the same root are lost there is always a violation of the Bell inequality corresponding to the induced graph, and in many cases also with respect to the full graph. The exact conditions are summarized in Table~\ref{table:conditions_for_violation}.

\begin{figure}[t]
	\begin{center}
		\begin{subfigure}{6cm}
		\centering
			\includegraphics[height=3cm]{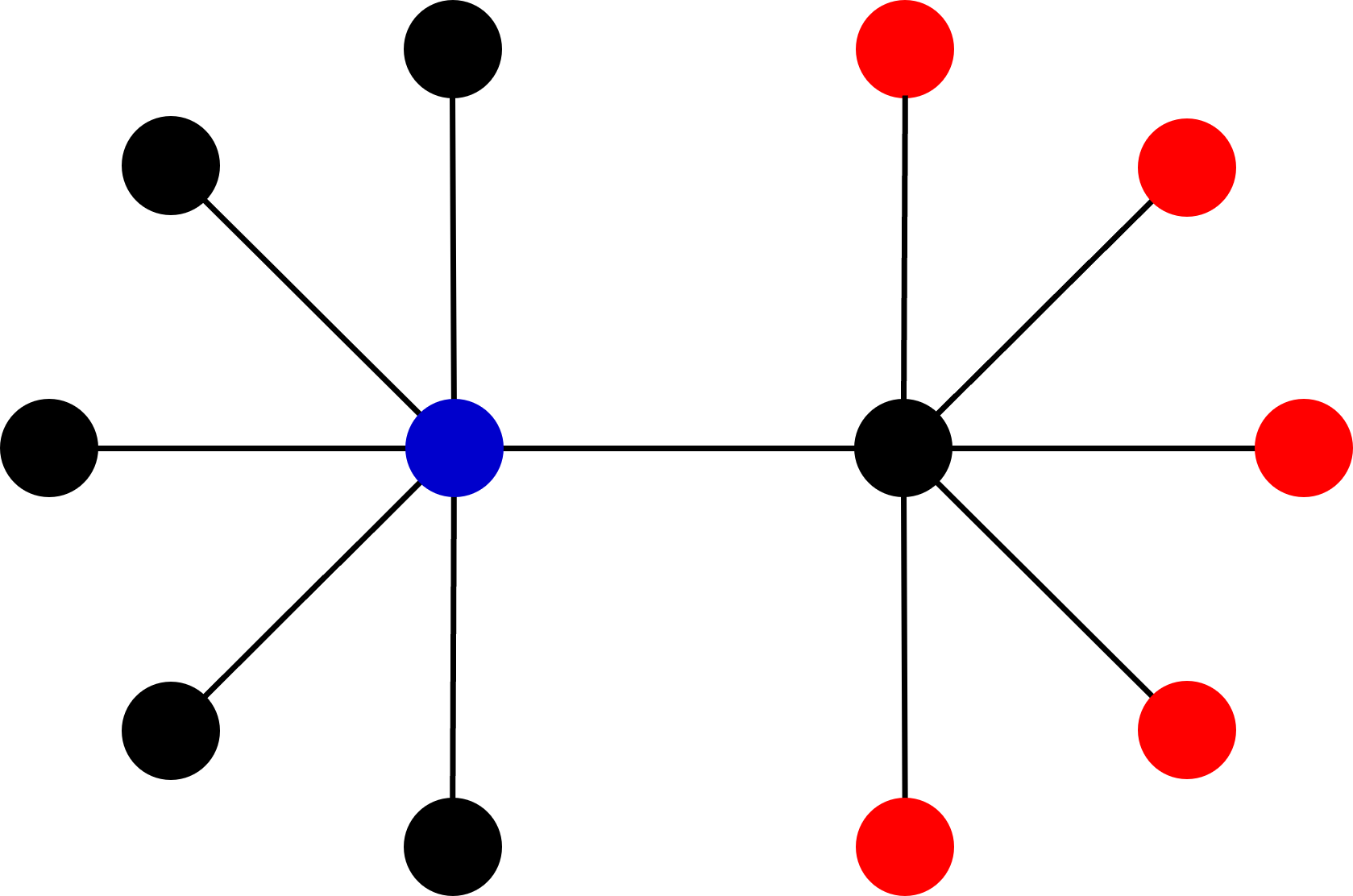}
			\caption{}
	    	\label{fig:two_centered_ghz}
		\end{subfigure}
		\quad
		\begin{subfigure}{6cm}
		\centering
			\includegraphics[height=3cm]{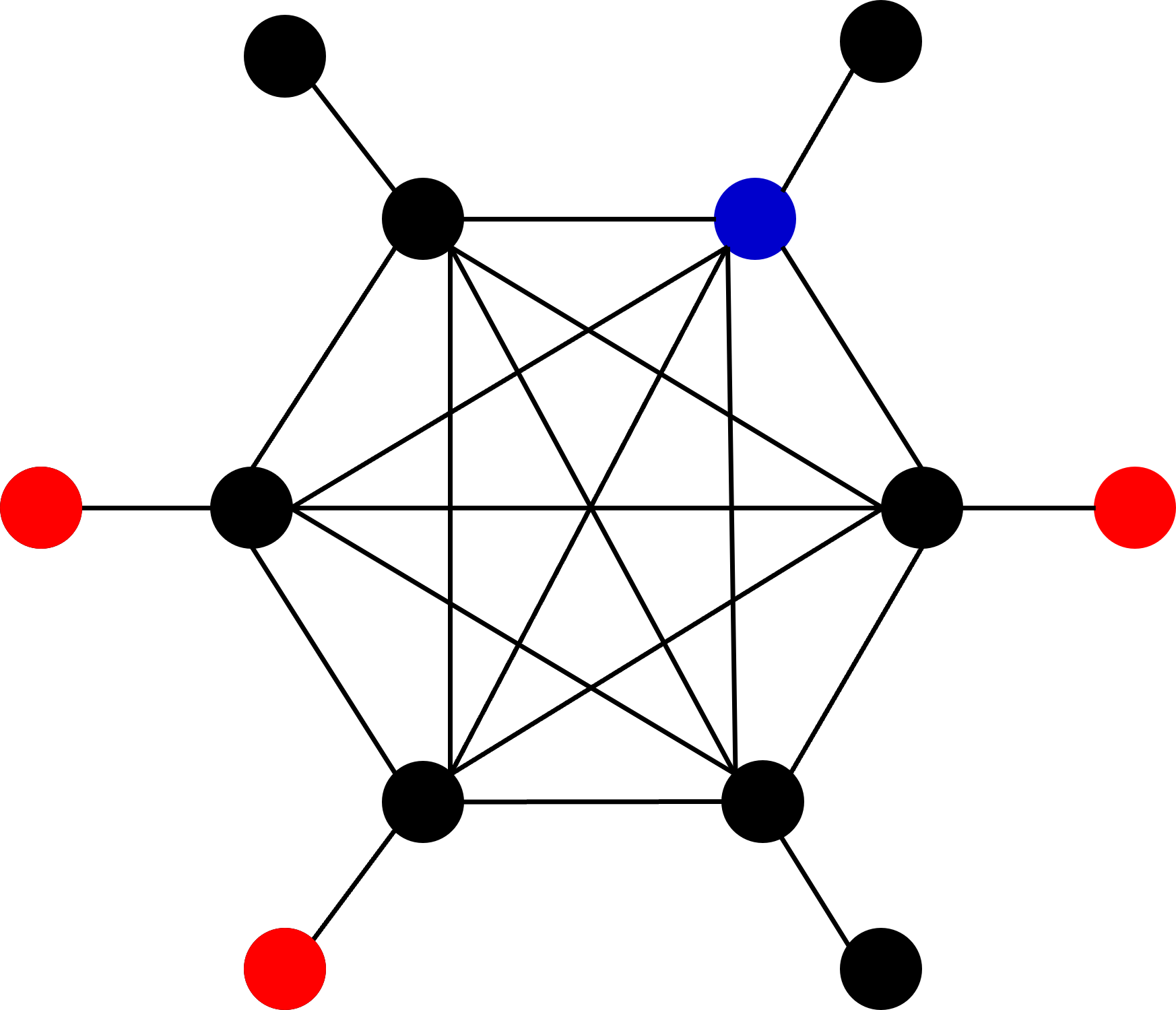}
			\caption{}
			\label{fig:complete_like_graph_state}
		\end{subfigure}
	\end{center}
	\caption{Loss tolerant graph examples with \(N=12\) which violate the classical limit with respect to the induced graph Bell operator \(I^{G*}_{r}\). The chosen root vertex is colored in blue, and the lost qubits are colored in red. (\subref{fig:two_centered_ghz})  A ``two centered GHZ''. In this case, the root vertex marked in blue is the only choice. All the qubits adjacent to the right root are lost and the violation is maintained. If qubits adjacent to both roots are lost there is no violation. (\subref{fig:complete_like_graph_state}) A ``dense center graph state''. Any three leaves can be lost while maintaining the violation of the classical limit. For any combination there are three choices of root.}
\end{figure}

The two centered GHZ example demonstrates that having more than one root enables robustness to loss of qubits. If a qubit adjacent to one of the roots is lost, it is possible to evaluate the Bell inequality with respect to the other root and maintain the violation. However, it allowed loss of qubits adjacent to only one of the roots. This can be improved upon by constructing graphs with a larger number of roots which will violate the Bell inequality as long as all neighbors of some root aren't lost. This leads to the following example.

\paragraph{``Dense center graph state''.} Consider the graph in  Figure~\ref{fig:complete_like_graph_state} \citep{Azuma_2015}. It is constructed of a complete graph in the middle and a leaf connected to each vertex in the complete graph. The classical bound in this case is given by \(\beta^{\text{Dense center}}_{\text{C}} = \frac{3}{2}N - 1\). 
As before, if one of the qubits from the complete-graph center is lost, there is no violation. 


Remarkably, for this graph, it is possible to violate the classical limit when several leaves are lost. The maximal number of lost qubits for which the remaining state still violates the classical limit can be calculated as follows. Every lost leaf qubit decreases the original quantum limit by \(\sqrt{2} + 1\), i.e,
\begin{align}
    \left\langle I^{\text{Dense center*}}_{r} \right\rangle_{\rho}= \left\langle I^{\text{Dense center'*}}_{r} \right\rangle_{\rho} = \beta_{\text{Q}}^{\text{Dense center}} -\left(\sqrt{2} + 1\right)\left|\mathcal{L}\right|
     \;.
\end{align}
As a result, the violation is maintained as long as \(\left|\mathcal{L}\right| \leq \frac{\sqrt{2} - 1}{\sqrt{2}} N\) for the induced graph and \(\left|\mathcal{L}\right| \leq \frac{\sqrt{2} - 1}{\sqrt{2} + 1} N\) for the full graph. The results are summarized in Table~\ref{table:conditions_for_violation}. Note that the complete graph in the center can be approximated by some other, lower degree \(k-\text{regular}\) graph. In this case, since \(n_{\text{max}}\) is smaller the amount of lost qubits allowed will be smaller.

Although a smaller number of leaves can be lost in the dense center example compared with the two-centered GHZ example, the dense center graph has a better performance in terms of the number of leaves that can be lost in the worst case with respect to which set of leaves is lost. We remark that these two examples are complementary is some sense.  The difference between the classical and the quantum bound is linear in \(n_{\text{max}}\) and hence the larger \(n_{\text{max}}\) is, the better. In order for a violation to hold when qubits are lost, the graph topology should protect \(n_{\text{max}}\) from changing. In the two-centered GHZ example \(n_\text{max}\) is large by making each root connected to many leaves while in the dense center example this is achieved by the roots being highly connected to one another. 

\begin{table}
\centering
\begin{tabular}{ | P{2cm} | P{4.2cm}| P{4.2cm} | P{4.2cm}| }
 \hline
 Graph & 
 Expectation value of \(\left\langle I^{G*}_{r} \right\rangle_{\rho}, \left\langle I^{G'*}_{r} \right\rangle_{\rho}\) & 
 Maximal number of lost qubits with respect to \(I^{G'}_{r}\) &
 Comments \\ [0.5ex] 
 
 \hline\hline
 Two-centered GHZ & 
 \(\left(2\sqrt{2} - 1\right)n_{\text{max}} + N - 1 -\sqrt{2} - \left|\mathcal{L}\right|\) & 
 \(\left|\mathcal{L}\right| \leq \frac{N}{2} - 1\) &
 Qubits adjacent to only one root can be lost \\
 
 \hline
 Dense center & 
 \(\left(2\sqrt{2} - 1\right)n_{\text{max}} + N-1 -\left(\sqrt{2} + 1\right)\left|\mathcal{L}\right|\)  &
 \(\left|\mathcal{L}\right| \leq \lfloor\frac{\sqrt{2} - 1}{\sqrt{2}}N\rfloor\) &
 Only leaf qubits can be lost \\
 
  \hline
\end{tabular}

\caption{Characterization of the expectation value of the Bell operator in Equation~\eqref{eqn:bell_operator_as_a_sum_of_stabilizing_operators} as a function of number of lost qubits together with the maximal number of qubits which can be lost while maintaining the violation of the classical limit for the two-centered GHZ (for \(1 \leq \left| \mathcal{L} \right| \leq \frac{N}{2} - 1\)) and the dense center (for \(0 \leq \left| \mathcal{L} \right| \leq \frac{N}{2} - 1\)) graphs.}
\label{table:conditions_for_violation}
\end{table}

\subsubsection{Unknown Set of Lost Qubits}
\label{unkown_set_of_lost_qubits_section}
The previous sections discussed situations in which the set of lost qubits is known. This section is devoted to the scenario where this information is unavailable. 
In this case, the density matrix is a convex combination of density matrices in the form of Equation~\eqref{eqn:remaining_density_matrix}, corresponding to all the possibilities of lost qubit sets and their probability. Thus, the expectation value of the Bell operator corresponding to the full graph should be calculated in a similar way as in the previous sections. It is the linear combination of all the expectation values as in Theorem~\ref{expectation_value_theorem} with coefficients according to the probability for each density matrix in the convex sum. 
It is also natural to ask how the expectation value with respect to one of the possible induced graphs behaves in these cases. When the lost qubits were known, it was always better to evaluate the expectation value of the Bell operator of the induced graph. Interestingly, we find that this is no longer clear cut when the set of lost qubits is unknown. This should be taken into account when deciding which measurements to perform.

For example, suppose that in the dense center graph example described in Section~\ref{loss_tolerant_section} and in Figure~\ref{fig:complete_like_graph_state} there is probability \(p\) for a single qubit loss of one of the leaves (without knowing which). The density matrix is a sum of the following density matrices: \(\left(1-p\right)\left|\phi^{G}\right\rangle \left\langle\phi^{G}\right| + \frac{2p}{N}\sum_{\mathcal{L}}\rho_\mathcal{L}\) (where \(\rho_\mathcal{L}\) is the density matrix corresponding to each possibility of lost qubit). For every choice of root there is a density matrix in the sum corresponding to the loss of the leaf adjacent to the chosen root. 
For the full graph Bell operator, as seen before, when this leaf is lost, there is no violation. However, the total state still violates the Bell inequality by a margin which decreases with larger \(p\). In the limit of small \(p\) the violation approaches its maximal possible margin.
On the other hand, choosing to focus on the Bell operator of one of the possible induced graphs, the violation in the limit \(p \rightarrow 0\) is dominated by the case where no qubit is lost, and as shown in Appendix~\ref{appendix:full_graph_with_sub_graph} in Lemma~\ref{full_graph_with_sub_graph_operator}, the violation is not optimal since the expectation value of some stabilizers in the sum vanishes. Generally, this advantage of the full graph Bell inequality should be expected to hold up to some finite value of \(p\). 

\section{Discussion and Open Questions}
In this work we have discussed the robustness of Bell violation of graph states to loss of qubits. We used the stabilizer formalism to derive a lemma which, for any given graph state, provides the exact
decrease in the Bell violation for any subset of lost qubits. Provided a graph state, the results allow to determine how many qubits can be lost while still maintaining the violation. In addition, it allows to identify qubits that are more important and must be kept safe since their loss demolishes the violation. 
In particular, it was shown that the widely studied GHZ and ring
topologies are susceptible to loss of any qubit, while other topologies such as the ``two-centered GHZ'' and ``dense center'' graphs presented in Section~\ref{loss_tolerant_section} are more resilient. 

We find that for no-click events (scenarios in which, e.g., photon detectors are expected to detect a photon but they do not and thus it is known which qubits were lost) considering the Bell operator of the induced sub-graph of the remaining qubits better optimizes the violation. The expectation value of the Bell operator of the full graph and the induced graph in this case turn out to be the same, but the classical bound for the induced sub-graph is lower and hence the violation is more significant. In addition, we've discussed the scenario in which the information of which qubit is lost is unavailable. For these cases, we've demonstrated examples in which it is better to use the Bell operator corresponding to the full graph.  

For future work, the main lemma derived here can be used to optimize graph topologies with respect to different metrics of robustness. The relevant metric should arise from constraints in real experimental setups and depend on use-case specification for the graph state. For example, one could consider maximal loss of qubits where there is still a violation, maximal violation after a specified subset of qubits is lost, and so on. Furthermore, other sources of error such as bit flips and phase errors may be taken into account in the analysis. In addition, the stabilizer formalism approach used here to derive Lemma~\ref{lost_qubits_lemma} can be applied to other Bell inequalities violated by graph states as they are all sums of multiplications of the stabilizing operators.

\section{Acknowledgments}
We would like to thank Ziv Aqua and Barak Dayan for insightful discussions and Thomas Vidick for helpful comments on the manuscript. 
This research was generously supported by the Peter and Patricia Gruber Award, the Daniel E. Koshland Career Development Chair, the Koshland Research Fund and the Minerva foundation with funding from the Federal German Ministry for Education and Research.

\appendix

\section{Proof of Lemma \ref{lost_qubits_lemma}} 

\universe*

\begin{proof}
\label{appendix:full_lemma_proof}
	Express the graph states as follows
	
	\begin{align}
		\label{eqn:rewriting_phi_G}
		\left|\phi^{G}\right\rangle =
		\sum_{s \in \left\{0,1\right\}^{\left|\mathcal{L}\right|}} \alpha_{s}  \left|\psi_{s}\right\rangle_{\left\{1,\dots,N\right\}\setminus\mathcal{L}} \otimes \left|s\right\rangle_{\mathcal{L}}
		 \;,
	\end{align}
	
	where \(\left\{\left|\psi_{s}\right\rangle_{\left\{1,\dots,N\right\}\setminus\mathcal{L}}\right\} _{s \in \left\{0,1\right\}^{\left|\mathcal{L}\right|}}\) are normalized in the remaining \(N - \left|\mathcal{L}\right|\) dimensional Hilbert space. In addition,
	\begin{align}
		\sum_{s \in \left\{ 0,1 \right\}^{\left|\mathcal{L}\right|}} \left|\alpha_{s}\right|^{2}=1 
		 \;,
	\end{align}

	since \(\left|\phi^{G}\right\rangle\) is normalized. Hence 
	
	\begin{align}
		\rho & =\mathrm{Tr}_{\mathcal{L}} \left[ 
		\left(\sum_{s \in \left\{0,1\right\}^{\left|\mathcal{L}\right|}} \alpha_{s}  \left|\psi_{s}\right\rangle_{\left\{1,\dots,N\right\}\setminus\mathcal{L}} \otimes \left|s\right\rangle_{\mathcal{L}} \right)
		\left(\sum_{s \in \left\{0,1\right\}^{\left|\mathcal{L}\right|}} \alpha_{s}^{*}  \left\langle\psi_{s}\right|_{\left\{1,\dots,N\right\}\setminus\mathcal{L}} \otimes \left\langle s \right|_{\mathcal{L}}
		\right) \right]
		\otimes \left(\bigotimes_{l \in \mathcal{L}} \left|0 \right\rangle_{l} \left\langle 0 \right|_{l} \right) \nonumber \\ 
		& =
		\left(
		\sum_{s \in \left\{0,1\right\}^{\left|\mathcal{L}\right|}} \left|\alpha_{s}\right|^{2}  \left|\psi_{s}\right\rangle_{\left\{1,\dots,N\right\}\setminus\mathcal{L}}  \left\langle\psi_{s}\right|_{\left\{1,\dots,N\right\}\setminus\mathcal{L}}\right)
		\otimes \left(\bigotimes_{l \in \mathcal{L}} \left|0 \right\rangle_{l} \left\langle 0 \right|_{l} \right)
		 \;.
	\end{align}
	
	The expectation values in Equation~\eqref{eqn:stablizer_expectation_value} are calculated separately for each  of the following cases
	
	\begin{enumerate}
		\item \(i \notin \bigcup_{l \in \mathcal{L}} \overline{\mathcal{N}_{l}^{G}}\)
		\item \(i\in \left(\bigcup_{l \in \mathcal{L}} \mathcal{N}_{l}^{G} \right)\setminus \mathcal{L}\)
		\item \(i \in \mathcal{L}\) (Relevant only for part a)
	\end{enumerate}
	
	The proof is divided into two parts. In Part A the eigenvalue equation (Equation~\eqref{eqn:eigenvalue_equation}) is used in order to calculate the action of the stabilizers \(S_{i}^{G}\) omitting their support on the lost qubits, on the states \(\left|\psi_{s}\right\rangle\) for all \(i \in \left\{1,\dots,N\right\}, s \in \left\{0,1\right\}^{\left|\mathcal{L}\right|}\). These reduced operators are in fact the stabilizers of the induced graph \(S_{i}^{G\left[V\setminus \right\{ v_{k} \text{ s.t } k\in\mathcal{L} \left\}\right]}\) for \(i\in \left\{1\dots N\right\}\setminus\mathcal{L}\). In Part B the result is applied in order to calculate the expectation value of the stabilizing operators \(S_{i}^{G},S_{i}^{G\left[V\setminus \right\{ v_{k} \text{ s.t } k\in\mathcal{L} \left\}\right]}\). Note that when stabilizers act on states \(\left|\psi_{s}\right\rangle\) in a reduced Hilbert space, it is to be understood that the identity operators in the complementary system are discarded.

	\paragraph{Part A: Calculating the Action of the Stabilizers on \(\left|\psi_{s}\right\rangle\).}
	\label{part_a}
	\begin{enumerate}
		\item For \(i \notin \bigcup_{l \in \mathcal{L}} \overline{\mathcal{N}_{l}^{G}}\): \(S_{i}^{G} = S_{i}^{G\left[V\setminus \right\{ v_{k} \text{ s.t } k\in\mathcal{L} \left\}\right]}\).
		Using the eigenvalue equation
		
		\begin{align}
			S_{i}^{G}\left|\phi^{G}\right\rangle & = 
			\sum_{s \in \left\{0,1\right\}^{\left|\mathcal{L}\right|}} \alpha_{s} S_{i}^{G} \left|\psi_{s}\right\rangle_{\left\{1,\dots,N\right\}\setminus\mathcal{L}} \otimes \left|s\right\rangle_{\mathcal{L}}
			\nonumber\\ & =
			\sum_{s \in \left\{0,1\right\}^{\left|\mathcal{L}\right|}} \alpha_{s} S_{i}^{G\left[V\setminus \right\{ v_{k} \text{ s.t } k\in\mathcal{L} \left\}\right]} \left|\psi_{s}\right\rangle_{\left\{1,\dots,N\right\}\setminus\mathcal{L}} \otimes \left|s\right\rangle_{\mathcal{L}} \nonumber\\ & =
			\sum_{s \in \left\{0,1\right\}^{\left|\mathcal{L}\right|}} \alpha_{s}  \left|\psi_{s}\right\rangle_{\left\{1,\dots,N\right\}\setminus\mathcal{L}} \otimes \left|s\right\rangle_{\mathcal{L}}
			 \;.
		\end{align}
		
		For all \( \left|s\right\rangle_{\mathcal{L}} \), \(s \in \left\{0,1\right\}^{\left|\mathcal{L}\right|}\)
		\begin{align}
			\label{eqn:eigenvalue_equation_non_neighbors}
			S_{i}^{G}  \left|\psi_{s}\right\rangle_{\left\{1,\dots,N\right\}\setminus\mathcal{L}} = 
			S_{i}^{G\left[V\setminus \right\{ v_{k} \text{ s.t } k\in\mathcal{L} \left\}\right]} \left|\psi_{s}\right\rangle_{\left\{1,\dots,N\right\}\setminus\mathcal{L}} = 
			\left|\psi_{s}\right\rangle_{\left\{1,\dots,N\right\}\setminus\mathcal{L}}
			 \;,
		\end{align}
		meaning that \(\left\{\left|\psi_{s}\right\rangle_{\left\{1,\dots,N\right\}\setminus\mathcal{L}}\right\} _{s \in \left\{0,1\right\}^{\left|\mathcal{L}\right|}}\) are eigenstates of \(S_{i}^{G}  = S_{i}^{G\left[V\setminus \right\{ v_{k} \text{ s.t } k\in\mathcal{L} \left\}\right]}\) with eigenvalue \(1\).
		
		\item For  \(i\in \left( \bigcup_{l \in \mathcal{L}} \mathcal{N}_{l}^{G} \right)\setminus \mathcal{L}\) 
		\begin{align}
		    S_{i}^{G} = S_{i}^{G\left[V\setminus \right\{ v_{k} \text{ s.t } k\in\mathcal{L} \left\}\right]} \bigotimes_{l\in \mathcal{N}_{i}^{G} \cap \mathcal{L}} Z_{l}
		     \;.
		\end{align}
		Hence, 
		
		\begin{align}
			S_{i}^{G}\left|\phi^{G}\right\rangle & = S_{i}^{G\left[V\setminus \right\{ v_{k} \text{ s.t } k\in\mathcal{L} \left\}\right]} \bigotimes_{l\in \mathcal{N}_{i}^{G} \cap \mathcal{L}} Z_{l}  \sum_{s \in \left\{0,1\right\}^{\left|\mathcal{L}\right|}} \alpha_{s}  \left|\psi_{s}\right\rangle_{\left\{1,\dots,N\right\}\setminus\mathcal{L}} \otimes \left|s\right\rangle_{\mathcal{L}}  \nonumber \\ 
			& =  \sum_{s \in \left\{0,1\right\}^{\left|\mathcal{L}\right|}} \alpha_{s} \left( S_{i}^{G\left[V\setminus \right\{ v_{k} \text{ s.t } k\in\mathcal{L} \left\}\right]} \left|\psi_{s}\right\rangle_{\left\{1,\dots,N\right\}\setminus\mathcal{L}} \right) \otimes \left( \bigotimes_{l\in \mathcal{N}_{i}^{G} \cap \mathcal{L}} Z_{l} \left|s\right\rangle_{\mathcal{L}}  \right)  \nonumber \\
			& = \sum_{s \in \left\{0,1\right\}^{\left|\mathcal{L}\right|}} \alpha_{s} \left( S_{i}^{G\left[V\setminus \right\{ v_{k} \text{ s.t } k\in\mathcal{L} \left\}\right]} \left|\psi_{s}\right\rangle_{\left\{1,\dots,N\right\}\setminus\mathcal{L}} \right) \otimes \left[ \bigotimes_{l\in \mathcal{N}_{i}^{G} \cap \mathcal{L}} Z_{l} \left( \bigotimes_{l'\in\mathcal{L}} \left|s_{l'}\right\rangle_{l'} \right) \right]  
			\nonumber \\ 
			& = \sum_{s \in \left\{0,1\right\}^{\left|\mathcal{L}\right|}} \alpha_{s} \left( S_{i}^{G\left[V\setminus \right\{ v_{k} \text{ s.t } k\in\mathcal{L} \left\}\right]} \left|\psi_{s}\right\rangle_{\left\{1,\dots,N\right\}\setminus\mathcal{L}} \right) \otimes \left( \left(-1\right)^{\sum_{l\in\mathcal{L} \cap \mathcal{N}_{i}^{G}} s_{l}} \left|s\right\rangle_{\mathcal{L}}  \right) 
			\nonumber \\ & =
			\sum_{s \in \left\{0,1\right\}^{\left|\mathcal{L}\right|}} \alpha_{s}  \left|\psi_{s}\right\rangle_{\left\{1,\dots,N\right\}\setminus\mathcal{L}} \otimes \left|s\right\rangle_{\mathcal{L}}
			 \;,
		\end{align}
		where in the third step \( \left|s\right\rangle = \bigotimes_{l'\in\mathcal{L}} \left|s_{l'}\right\rangle_{l'}\). 
		For all \( \left|s\right\rangle_{\mathcal{L}} \), \(s \in \left\{0,1\right\}^{\left|\mathcal{L}\right|}\)
		
		\begin{align}
			\label{eigenvalue_equation_neighbors}
			S_{i}^{G\left[V\setminus \right\{ v_{k} \text{ s.t } k\in\mathcal{L} \left\}\right]} \left|\psi_{s}\right\rangle_{\left\{1,\dots,N\right\}\setminus\mathcal{L}} = \left(-1\right)^{\sum_{l\in\mathcal{L} \cap \mathcal{N}_{i}^{G}} s_{l}} \left|\psi_{s}\right\rangle_{\left\{1,\dots,N\right\}\setminus\mathcal{L}}
			 \;,
		\end{align}
		meaning that \(\left|\psi_{s}\right\rangle_{\left\{1,\dots,N\right\}\setminus\mathcal{L}}\) are eigenstates of \(S_{i}^{G\left[V\setminus \right\{ v_{k} \text{ s.t } k\in\mathcal{L} \left\}\right]}\) with eigenvalues \(-1,1\). The eigenvalue depends on the parity of the number of \(1\)'s in the string \(s \in \left\{ 0,1 \right\}^{\left|\mathcal{L}\right|}\) in indices \( l \in \mathcal{L} \cap \mathcal{N}_{i}^{G}\). 
		
		The number of states \(\left|\psi_{s}\right\rangle\) with eigenvalue \(1\) is equal to the number of states with eigenvalue \(-1\). This can easily be shown by proving a one-to-one correspondence between these two sets. Choose \(l' \in \mathcal{L} \cap \mathcal{N}_{i}^{G}\) (it is possible since \(\mathcal{L} \cap \mathcal{N}_{i}^{G}\) is not empty). Notice that \(s = s_{l_{0}},\dots,s_{l'},\dots ,s_{l_{\left|\mathcal{L}\right|}}\) and \(s = s_{l_{0}},\dots,\bar{s_{l'}},\dots s_{l_{\left|\mathcal{L}\right|}}\) have opposite parity and therefore have opposite signs of their eigenvalues with respect to \(S_{i}^{G\left[V\setminus \right\{ v_{k} \text{ s.t } k\in\mathcal{L} \left\}\right]}\).
		
		\item For \(i \in \mathcal{L}\) using the eigenvalue equation as before
		
		\begin{align}
			S_{i}^{G} \left|\phi^{G}\right\rangle & = X_{i}\bigotimes_{j\in \mathcal{N}_{i}^{G} \setminus \mathcal{L}}Z_{j} \bigotimes_{l\in \mathcal{N}_{i}^{G} \cap \mathcal{L}} Z_{l} \sum_{s \in \left\{0,1\right\}^{\left|\mathcal{L}\right|}} \alpha_{s}  \left|\psi_{s}\right\rangle_{\left\{1,\dots,N\right\}\setminus\mathcal{L}} \otimes \left|s\right\rangle_{\mathcal{L}} \nonumber \\
			& = 
			\sum_{s \in \left\{0,1\right\}^{\left|\mathcal{L}\right|}} \alpha_{s} \left( \bigotimes_{j\in \mathcal{N}_{i}^{G} \setminus \mathcal{L}} Z_{j} \left|\psi_{s}\right\rangle_{\left\{1,\dots,N\right\}\setminus\mathcal{L}} \right) \otimes \left( X_{i} \otimes \bigotimes_{l\in \mathcal{N}_{i}^{G} \cap \mathcal{L}} Z_{l} \left|s\right\rangle_{\mathcal{L}} \right)  \nonumber \\
			& = 
			\sum_{s \in \left\{0,1\right\}^{\left|\mathcal{L}\right|}} \alpha_{s} \left( \bigotimes_{j\in \mathcal{N}_{i}^{G} \setminus \mathcal{L}} Z_{j} \left|\psi_{s}\right\rangle_{\left\{1,\dots,N\right\}\setminus\mathcal{L}} \right) \otimes \left[ X_{i} \otimes \bigotimes_{l\in \mathcal{N}_{i}^{G} \cap \mathcal{L}} Z_{l} \left(\bigotimes_{l'\in\mathcal{L}} \left|s_{l'}\right\rangle_{l'}\right) \right]  \nonumber \\
			& = 
			\sum_{s \in \left\{0,1\right\}^{\left|\mathcal{L}\right|}} \alpha_{s} \left( \bigotimes_{j\in \mathcal{N}_{i}^{G} \setminus \mathcal{L}} Z_{j} \left|\psi_{s}\right\rangle_{\left\{1,\dots,N\right\}\setminus\mathcal{L}} \right) \otimes \left[ \left(X_{i}\left|s_{i}\right\rangle_{i} \right) \otimes \bigotimes_{l\in \mathcal{N}_{i}^{G} \cap \mathcal{L}} Z_{l} \left(\bigotimes_{l'\in\mathcal{L} \setminus \left\{i\right\}} \left|s_{l'}\right\rangle_{l'}\right) \right]   \nonumber \\
			& = 
			\sum_{s \in \left\{0,1\right\}^{\left|\mathcal{L}\right|}} \alpha_{s} \left( \bigotimes_{j\in \mathcal{N}_{i}^{G} \setminus \mathcal{L}} Z_{j} \left|\psi_{s}\right\rangle_{\left\{1,\dots,N\right\}\setminus\mathcal{L}} \right) \otimes \left[\left(-1\right)^{\sum_{l\in\mathcal{L} \cap \mathcal{N}_{i}^{G}} s_{l}}
			\left|\bar{s_{i}}\right\rangle_{i}  \otimes \left(\bigotimes_{l'\in\mathcal{L} \setminus \left\{i\right\}} \left|s_{l'}\right\rangle_{l'}\right) \right]   \nonumber \\  
			& = \sum_{s \in \left\{0,1\right\}^{\left|\mathcal{L}\right|}} \alpha_{s}  \left|\psi_{s}\right\rangle_{\left\{1,\dots,N\right\}\setminus\mathcal{L}} \otimes \left|s\right\rangle_{\mathcal{L}}
			 \;.
		\end{align}
		
		Again, in the third step \( \left|s\right\rangle = \bigotimes_{l'\in\mathcal{L}} \left|s_{l'}\right\rangle_{l'}\).
		For all \( \left|s\right\rangle_{\mathcal{L}} \), \(s \in \left\{0,1\right\}^{\left|\mathcal{L}\right|}, s = s_{l_{0}},\dots,s_{l_{i}},\dots ,s_{l_{\left|\mathcal{L}\right|}}\)
		
		\begin{align}
			\label{eqn:equal_weight_property}
			\alpha_{s_{l_{0}},\dots,s_{l_{i}},\dots ,s_{l_{\left|\mathcal{L}\right|}}} \bigotimes_{j\in \mathcal{N}_{i}^{G} \setminus \mathcal{L}} Z_{j} \left|\psi_{s_{l_{0}},\dots,s_{l_{i}},\dots ,s_{l_{\left|\mathcal{L}\right|}}}\right\rangle_{\left\{1,\dots,N\right\}\setminus\mathcal{L}} = \\ \left(-1\right)^{\sum_{l\in\mathcal{L} \cap \mathcal{N}_{i}^{G}} s_{l}} \alpha_{s_{l_{0}},\dots,\bar{s_{l_{i}}},\dots ,s_{l_{\left|\mathcal{L}\right|}}}\left|\psi_{s_{l_{1}},\dots\bar{s_{l_{i}}},\dots_{l_{\left|\mathcal{L}\right|}}}\right\rangle_{\left\{1,\dots,N\right\}\setminus\mathcal{L}}
			 \;.
		\end{align}
		
		Notice that for all \(i \in \mathcal{L}\)
		\begin{align}
		    \left|\alpha_{s_{l_{0}},\dots,s_{l_{i}},\dots ,s_{l_{\left|\mathcal{L}\right|}}}\right|= \left|\alpha_{s_{l_{0}},\dots,\bar{s_{l_{i}}},\dots,s_{l_{\left|\mathcal{L}\right|}}}\right|
		    \;,
		\end{align}
		since the operator \(\bigotimes_{j\in \mathcal{N}_{i}^{G} \setminus \mathcal{L}} Z_{j}\) is unitary, meaning that the states remain normalized after it is applied.
		
		Hence, for all \(s,s' \in \left\{0,1\right\}^{\left|\mathcal{L}\right|}\)
		\begin{align}
		\label{eqn:coefficient_equality}
		    \left|\alpha_{s}\right| =\left|\alpha_{s'}\right|
		     \;,
		\end{align} 
		since \(s,s'\) differ from one another by flipping a finite set of indices, and one can create a series of equalities  transforming from \(\left|\alpha_{s}\right|\) to  \(\left|\alpha_{s'}\right|\) by flipping one index at a time while maintaining the equality.
	\end{enumerate}
	
	\paragraph{Part B: Calculating the Expectation Values.}
	\begin{enumerate}
		\item For \(i \notin \bigcup_{l \in \mathcal{L}} \overline{\mathcal{N}_{l}^{G}}\), using the eigenvalue equation obtained in Equation~\eqref{eqn:eigenvalue_equation_non_neighbors}
		
		\begin{align}
			\left\langle S_{i}^{G} \right\rangle_{\rho} & = 
			\mathrm{Tr}\left[ \left( S_{i}^{G}    
			\sum_{s \in \left\{0,1\right\}^{\left|\mathcal{L}\right|}} \left|\alpha_{s}\right|^{2}  \left|\psi_{s}\right\rangle_{\left\{1,\dots,N\right\}\setminus\mathcal{L}}  \left\langle\psi_{s}\right|_{\left\{1,\dots,N\right\}\setminus\mathcal{L}}\right)
			\otimes \left( \bigotimes_{l \in \mathcal{L}} \left|0 \right\rangle_{l} \left\langle 0 \right|_{l} \right) \right] \nonumber\\ & = 
			\mathrm{Tr}_{\left\{ 1,\dots,N\right\} \setminus \mathcal{L}}\left[
			\sum_{s \in \left\{0,1\right\}^{\left|\mathcal{L}\right|}} \left|\alpha_{s}\right|^{2}  S_{i}^{G} \left|\psi_{s}\right\rangle_{\left\{1,\dots,N\right\}\setminus\mathcal{L}}  \left\langle\psi_{s}\right|_{\left\{1,\dots,N\right\}\setminus\mathcal{L}}\right] \nonumber \\ 
			& =
			\mathrm{Tr}_{\left\{ 1,\dots,N\right\} \setminus \mathcal{L}}\left[
			\sum_{s \in \left\{0,1\right\}^{\left|\mathcal{L}\right|}} \left|\alpha_{s}\right|^{2} \left|\psi_{s}\right\rangle_{\left\{1,\dots,N\right\}\setminus\mathcal{L}}  \left\langle\psi_{s}\right|_{\left\{1,\dots,N\right\}\setminus\mathcal{L}}\right] \nonumber \\  & = 1
			 \;.
		\end{align}
		
		In the third step the qubits in \(\mathcal{L}\) are traced out, in the fourth step Equation~\eqref{eqn:eigenvalue_equation_non_neighbors} is used, and the last step is correct due to the normalization of \(\left|\phi^{G}\right\rangle\)  as in Equation~\eqref{eqn:rewriting_phi_G}. The exact same calculation holds for \(S_{i}^{G\left[V\setminus \right\{ v_{k} \text{ s.t } k\in\mathcal{L} \left\}\right]} \).
		
		\item For \(i\in \left( \bigcup_{l \in \mathcal{L}} \mathcal{N}_{l}^{G} \right)\setminus \mathcal{L}\)
		\begin{align}
			\left\langle S_{i}^{G} \right\rangle_{\rho} 
			& = 
			\left\langle  S_{i}^{G\left[V\setminus \right\{ v_{k} \text{ s.t } k\in\mathcal{L} \left\}\right]} \bigotimes_{l\in \mathcal{N}_{i}^{G} \cap \mathcal{L}} Z_{l} \right\rangle 
			\nonumber \\ & = 
			\mathrm{Tr} \left[ S_{i}^{G\left[V\setminus \right\{ v_{k} \text{ s.t } k\in\mathcal{L} \left\}\right]} \bigotimes_{l\in \mathcal{N}_{i}^{G} \cap \mathcal{L}} Z_{l} \left(\sum_{s \in \left\{0,1\right\}^{\left|\mathcal{L}\right|}} \left|\alpha_{s}\right|^{2}  \left|\psi_{s}\right\rangle_{\left\{1,\dots,N\right\}\setminus\mathcal{L}}  \left\langle\psi_{s}\right|_{\left\{1,\dots,N\right\}\setminus\mathcal{L}}\right)
			\otimes \left( \bigotimes_{l \in \mathcal{L}} \left|0 \right\rangle_{l} \left\langle 0 \right|_{l} \right) \right] \nonumber \\ 
			& = 
			\mathrm{Tr} \left[ \left( \sum_{s \in \left\{0,1\right\}^{\left|\mathcal{L}\right|}} \left|\alpha_{s}\right|^{2}  S_{i}^{G\left[V\setminus \right\{ v_{k} \text{ s.t } k\in\mathcal{L} \left\}\right]} \left|\psi_{s}\right\rangle_{\left\{1,\dots,N\right\}\setminus\mathcal{L}}  \left\langle\psi_{s}\right|_{\left\{1,\dots,N\right\}\setminus\mathcal{L}} \right)
			\otimes \left(\bigotimes_{l \in \mathcal{N}_{i}^{G} \cap \mathcal{L}} Z_{l} \bigotimes_{l \in \mathcal{L}} \left|0 \right\rangle_{l} \left\langle 0 \right|_{l} \right) \right] \nonumber \\
			& = 
			\mathrm{Tr}_{\left\{1,\dots,N\right\}\setminus\mathcal{L}} \left[\sum_{s \in \left\{0,1\right\}^{\left|\mathcal{L}\right|}} \left|\alpha_{s}\right|^{2}  \left(-1\right)^{\sum_{l\in\mathcal{L} \cap \mathcal{N}_{i}^{G}} s_{l}} \left|\psi_{s}\right\rangle_{\left\{1,\dots,N\right\}\setminus\mathcal{L}}  \left\langle\psi_{s}\right|_{\left\{1,\dots,N\right\}\setminus\mathcal{L}} \right] \nonumber \\
			& = 
			\sum_{s \in \left\{0,1\right\}^{\left|\mathcal{L}\right|}} \left(-1\right)^{\sum_{l\in\mathcal{L} \cap \mathcal{N}_{i}^{G}} s_{l}} \left|\alpha_{s}\right|^{2}
			 \;.
		\end{align}
		
		In the fourth step Equation~\eqref{eigenvalue_equation_neighbors} is used as well as the notation \(\left|s\right\rangle = \bigotimes_{l'\in\mathcal{L}} \left|s_{l'}\right\rangle_{l'}\).
	    In Part A for \(i\in \left( \bigcup_{l \in \mathcal{L}} \mathcal{N}_{l}^{G} \right)\setminus \mathcal{L}\), it was shown that there is a one-to-one correspondence between the eigenstates \(\left|\psi_{s}\right\rangle\) with eigenvalue \(1\) and the eigenstates with eigenvalue \(-1\). In addition, in Equation~\eqref{eqn:coefficient_equality} it was shown that the coefficients \(\alpha_{s}\) for all \(s\in\left\{0,1\right\}^{\left|\mathcal{L}\right|}\) have the same absolute value. As a consequence, the sum above vanishes.
		
		The exact same calculation holds for \(S_{i}^{G\left[V\setminus \right\{ v_{k} \text{ s.t } k\in\mathcal{L} \left\}\right]}\) discarding the \(\bigotimes_{l\in \mathcal{N}_{i}^{G} \cap \mathcal{L}} Z_{l}\) part in the calculation.
		
		\item  For \(i\in\mathcal{L}\)
		\begin{align}
			\left\langle S_{i}^{G} \right\rangle_\rho & = \left\langle X_{i} \bigotimes_{j\in \mathcal{N}_{i}^{G} \setminus \mathcal{L}} Z_{j} \bigotimes_{l\in \mathcal{N}_{i}^{G} \cap \mathcal{L}} Z_{l} \right\rangle \nonumber \\ 
			& =
			\mathrm{Tr} \left[ X_{i} \bigotimes_{j\in \mathcal{N}_{i}^{G} \setminus \mathcal{L}} Z_{j} \bigotimes_{l\in \mathcal{N}_{i}^{G} \cap \mathcal{L}} Z_{l} \left(
			\sum_{s \in \left\{0,1\right\}^{\left|\mathcal{L}\right|}} \left|\alpha_{s}\right|^{2}  \left|\psi_{s}\right\rangle_{\left\{1,\dots,N\right\}\setminus\mathcal{L}}  \left\langle\psi_{s}\right|_{\left\{1,\dots,N\right\}\setminus\mathcal{L}} \right)
			\otimes \left( \bigotimes_{l \in \mathcal{L}}  \left|0 \right\rangle_{l} \left\langle 0 \right|_{l}\right) \right] \nonumber \\ & = 
			\mathrm{Tr} \left[\left(\bigotimes_{j\in \mathcal{N}_{i}^{G} \setminus \mathcal{L}} Z_{j} \sum_{s \in \left\{0,1\right\}^{\left|\mathcal{L}\right|}} \left|\alpha_{s}\right|^{2}  \left|\psi_{s}\right\rangle_{\left\{1,\dots,N\right\}\setminus\mathcal{L}}  \left\langle\psi_{s}\right|_{\left\{1,\dots,N\right\}\setminus\mathcal{L}}\right)
			\otimes \left(\bigotimes_{l\in \mathcal{N}_{i}^{G} \cap \mathcal{L}} Z_{l}\bigotimes_{l \in \mathcal{L} \setminus \left\{i\right\}} \left|0 \right\rangle_{l} \left\langle 0 \right|_{l} \right) \otimes \left( X_{i}  \left|0 \right\rangle_{i} \left\langle 0 \right|_{i}\right)\right] \nonumber \\ 
			& =
			\mathrm{Tr} \left[\left(\bigotimes_{j\in \mathcal{N}_{i}^{G} \setminus \mathcal{L}} Z_{j} \sum_{s \in \left\{0,1\right\}^{\left|\mathcal{L}\right|}} \left|\alpha_{s}\right|^{2}  \left|\psi_{s}\right\rangle_{\left\{1,\dots,N\right\}\setminus\mathcal{L}}  \left\langle\psi_{s}\right|_{\left\{1,\dots,N\right\}\setminus\mathcal{L}}\right)
			\otimes \left(\bigotimes_{l\in \mathcal{N}_{i}^{G} \cap \mathcal{L}} Z_{l}\bigotimes_{l \in \mathcal{L}\setminus \left\{i\right\}} \left|0 \right\rangle_{l} \left\langle 0 \right|_{l} \right) \otimes \left|1 \right\rangle_{i} \left\langle 0 \right|_{i}\right] \nonumber \\
			& = 0
			 \;.
		\end{align}
		
		The expression vanishes due to the partial trace on the \(i\)'th qubit. \qedhere
	\end{enumerate}
\end{proof}

\section{Expectation Value of an Induced Graph Stabilizer on the Original State}
\label{appendix:full_graph_with_sub_graph}

\begin{lemma}
\label{full_graph_with_sub_graph_operator}
Let \(G, \mathcal{L}\) be as in Theorem~\ref{expectation_value_theorem}, \(\left|\phi^{G}\right\rangle\) the graph state corresponding to \(G\) and 
\(S_{i}^{G}\) as in Lemma~\ref{lost_qubits_lemma}. Then,
\begin{align}
\nonumber
    \left\langle S^{G\left[ V\setminus\left\{ v_{k} \text{ s.t } k\in\mathcal{L}\right\}\right]}_{i} \right\rangle_{\phi^{G}}  = \begin{cases}
        0 & i \in \bigcup_{l \in \mathcal{L}} \mathcal{N}_{l}^{G} \\
        1 & \text{otherwise}
      \end{cases}
\end{align}
\end{lemma}

\begin{proof}
The proof is very similar to the proof of Lemma~\ref{lost_qubits_lemma}.
Write \(\left|\phi^{G}\right\rangle\) as in Equation~\eqref{eqn:rewriting_phi_G} and have, 
    
\begin{align}
    \left\langle S^{G\left[ V\setminus\left\{ v_{k} \text{ s.t } k\in\mathcal{L}\right\}\right]}_{i} \right\rangle_{\phi^{G}} 
     & = 
     \left\langle \phi^{G} \right| S^{G\left[ V\setminus\left\{ v_{k} \text{ s.t } k\in\mathcal{L}\right\}\right]}_{i} \left|\phi^{G}\right\rangle 
     \nonumber \\ & =
     \left(\sum_{s' \in \left\{0,1\right\}^{\left|\mathcal{L}\right|}} \alpha_{s'}^{*}  \left\langle\psi_{s'}\right|_{\left\{1,\dots,N\right\}\setminus\mathcal{L}} \otimes \left\langle s' \right|_{\mathcal{L}} \right) S_{i}^{G\left[V\setminus \right\{ v_{k} \text{ s.t } k\in\mathcal{L} \left\}\right]} \left( \sum_{s \in \left\{0,1\right\}^{\left|\mathcal{L}\right|}} \alpha_{s}  \left|\psi_{s}\right\rangle_{\left\{1,\dots,N\right\}\setminus\mathcal{L}} \otimes \left|s\right\rangle_{\mathcal{L}} \right) 
     \nonumber \\ & = 
     \left(\sum_{s' \in \left\{0,1\right\}^{\left|\mathcal{L}\right|}} \alpha_{s'}^{*}  \left\langle\psi_{s'}\right|_{\left\{1,\dots,N\right\}\setminus\mathcal{L}} \otimes \left\langle s' \right|_{\mathcal{L}} \right)  \left( \sum_{s \in \left\{0,1\right\}^{\left|\mathcal{L}\right|}} \alpha_{s}  \left(-1\right)^{\sum_{l\in\mathcal{L} \cap \mathcal{N}_{i}^{G}} s_{l}} \left|\psi_{s}\right\rangle_{\left\{1,\dots,N\right\}\setminus\mathcal{L}} \otimes \left|s\right\rangle_{\mathcal{L}} \right) \nonumber \\ & = 	
     \sum_{s \in \left\{0,1\right\}^{\left|\mathcal{L}\right|}} \left(-1\right)^{\sum_{l\in\mathcal{L} \cap \mathcal{N}_{i}^{G}} s_{l}} \left|\alpha_{s}\right|^{2}
     \;.
\end{align}   

In the third step we used Equation~\eqref{eigenvalue_equation_neighbors}. As in the proof of Lemma~\eqref{lost_qubits_lemma} this sum vanishes when \(\mathcal{L} \cap \mathcal{N}_{i}^{G} \neq \emptyset\).
\end{proof}

\bibliography{references}
\bibliographystyle{unsrt}

\end{document}